%% file: two_producible_state.tex
\documentclass[pra,aps,superscriptaddress,floatfix,tightenlines,reprint]{revtex4-2}
\usepackage{amsthm}
\usepackage{amsmath}
\usepackage{amssymb}
\usepackage{amsfonts}
\usepackage{epsfig}
\usepackage{bm}
\usepackage{times}
\usepackage[colorlinks,linkcolor=blue,anchorcolor=blue,citecolor=blue,urlcolor=blue]{hyperref}
\usepackage{braket}
\usepackage{graphicx}
\usepackage{setspace}
\usepackage{subfigure}
\usepackage{enumerate}

\theoremstyle{definition}
\newtheorem{definition}{Definition}
\newtheorem*{definition*}{Definition}

\theoremstyle{plain}
\newtheorem{theorem}{Theorem}

\newtheorem{lemma}{Lemma}
\newtheorem{proposition}{Proposition}

\newtheorem*{theorem*}{Theorem}
\newtheorem*{corollary*}{Corollary}
\newtheorem*{lemma*}{Lemma}
\newtheorem*{proposition*}{Proposition}

\NewDocumentEnvironment{variant}{O{theorem} D(){} m}
    {\addtocounter{#1}{-1}%
    \expandafter\renewcommand\csname the#1\endcsname{\ref{#3}$'$}%
    \begin{#1}[#2]}
    {\end{#1}}

\NewDocumentEnvironment{appthm}{O{theorem} D(){} m}
    {\addtocounter{#1}{-1}%
    \expandafter\renewcommand\csname the#1\endcsname{\ref{#3}}%
    \begin{#1}[#2]}
    {\end{#1}}

\begin{document}

    \title{Generalized Entanglement of Purification Criteria for 2-Producible States in Multipartite Systems}

    \author{Tian-Ren Jin}
    \affiliation{Institute of Physics, Chinese Academy of Sciences, Beijing 100190, China}
    \affiliation{School of Physical Sciences, University of Chinese Academy of Sciences, Beijing 100049, China}

    \author{Yu-Ran Zhang}
    \email{yuranzhang@scut.edu.cn}
    \affiliation{School of Physics and Optoelectronics, South China University of Technology, Guangzhou 510640, China}

    \author{Heng Fan}
    \thanks{Corresponding author}
    \email{hfan@iphy.ac.cn}
    \affiliation{Institute of Physics, Chinese Academy of Sciences, Beijing 100190, China}
    \affiliation{School of Physical Sciences, University of Chinese Academy of Sciences, Beijing 100049, China}
    \affiliation{Beijing Academy of Quantum Information Sciences, Beijing 100193, China}
    \affiliation{Hefei National Laboratory, Hefei 230088, China}
    \affiliation{Songshan Lake Materials Laboratory, Dongguan 523808, China}
    \affiliation{CAS Center for Excellence in Topological Quantum Computation, UCAS, Beijing 100190, China}

    \begin{abstract}
        Multipartite entanglement has a much more complex structure than bipartite entanglement.
        A state that lacks generic multipartite entanglement is 2-producible, i.e., it can be written as a tensor product of at most 2-partite entangled states.
        Recently, it has been proved that a tripartite pure state is 2-producible if and only if the gap between the entanglement of purification and its lower bound vanishes.
        Here, we show that the entanglement of purification gap is insufficient to detect more than tripartite entanglement in $4$-partite stabilizer states.
        We then generalize entanglement of purification to the multipartite cases, and demonstrate that a multipartite pure state is 2-producible if and only if all the generalized entanglement of purification gaps vanish.
        The generalized entanglement of purification gap quantifies the quantum communication cost for redistributing one part of the system to the others, and also relates to the local recoverability of a multipartite state and the relative entropy between that state and 2-producible states.
        Moreover, we calculate the generalized entanglement of purification for states satisfying the general Schmidt decomposition, which implies that $4$-partite stabilizer states do not necessarily have a general Schmidt decomposition.
        Our results provide a quantitative characterization of multipartite entanglement in multipartite system, which will promote further investigations and understanding of multipartite entanglement.
    \end{abstract}

    \maketitle

    \section{Introduction}

        Bipartite pure states always admit the Schmidt decomposition, meaning all pure entangled states can be reversibly converted into maximally entangled states~\cite{PhysRevA.53.2046}.
        Consequently, entanglement of bipartite pure states can be completely characterized with entanglement entropy.
        For mixed bipartite states, several measures and criteria are used to detect entanglement~\cite{PhysRevA.54.3824,Hayden_2001,PhysRevA.54.3824,PhysRevA.57.1619,10.1063/1.1643788,10.1063/1.1498001,PhysRevLett.77.1413,PhysRevA.59.4206,chen2002matrix,Rudolph_2003}.
        In particular, bipartite entanglement of mixed states exhibits bound entanglement, from which no maximal entanglement can be distilled~\cite{PhysRevLett.80.5239}, but finite entanglement is required to produce it~\cite{PhysRevLett.86.5803}.
        This implies that bipartite entanglement of mixed states is irreversible, and there is no unique computable quantity governing all the conversion of entanglement~\cite{10.1145/780542.780545,10.5555/2011350.2011361,lami2023no}.

        In terms of purification,  bipartite entanglement of mixed states is characterized by tripartite entanglement of pure states.
        Multipartite entanglement is even more complex.
        The multipartite states, admitting a general Schmidt decomposition~\cite{PERES199516,PhysRevA.59.3336}, act as a maximally entangled state in the context of entanglement distillation~\cite{GISIN19981}.
        However, generally the multipartite pure state does not always have a Schmidt decomposition, such as the W state~\cite{PhysRevA.62.062314}.
        Therefore, for multipartite states, it requires a set of inequivalent maximally entangled states for the entanglement distillation~\cite{PhysRevA.62.062314,PhysRevLett.111.110502}.
        Furthermore, the so-called semiseparable state is $n$-partite entangled (for $n\geq 3$) but remains separable under any partition of the system into no more than $(n-1)$ parts~\cite{PhysRevLett.82.5385}.

        This complexity necessitates distinguishing generic multipartite entanglement from bipartite entanglement.
        The state without multipartite entanglement is a $2$-producible state, i.e., a tensor product of at most bipartite entangled states~\cite{guhne2005multipartite,balasubramanian2014multiboundary}.
        Recently, a connection between $2$-producibility and the entanglement of purification (EoP)~\cite{10.1063/1.1498001} has been established in tripartite systems~\cite{PhysRevLett.126.120501}.
        it is proved that a tripartite pure state is 2-producible (also called a triangle state) if and only if the EoP gap, defined as $g(A:B) = 2 E_p(A:B) - I(A:B)$, vanishes.
        This suggests that the EoP gap serves as a measure of tripartite entanglement.

        A natural question arises: Can the bipartite EoP gap be used to characterize $2$-producibility for general $n$-partite systems ($n>3$)?
        Surprisingly, the answer is negative. 
        In this work, we demonstrate that for $4$-partite random stabilizer states, the vanishing of bipartite EoP gaps for all pairs fails to guarantee that the state is 2-producible. 
        This limitation necessitates a generalization of the EoP framework.

        In this paper, we generalize the concept of EoP from tripartite systems to multipartite systems and introduce the generalized EoP gap. 
        We prove a structure theorem: a multipartite pure state is $2$-producible if and only if all the generalized EoP gaps vanish. 
        This provides a rigorous and quantitative criterion for detecting generic multipartite entanglement. 
        The generalized EoP gap quantifies the optimal quantum communication cost required to redistribute one part of the system to the others in a state redistribution protocol. 
        Furthermore, we establish that the gap relates to the local recoverability of a multipartite state and the distance between the state and the set of $2$-producible states in terms of relative entropy.
        We also calculate the generalized EoP for states admitting general Schmidt decompositions, including the GHZ state, which reveal that $4$-partite stabilizer states do not necessarily possess a general Schmidt decomposition, contrast with the tripartite case.

        The remainder of this paper is organized as follows. 
        In Sec.~\ref{sec: preliminaries}, we provide the necessary preliminaries on the EoP and the quantum Markov property. 
        In Sec.~\ref{sec: 2-producible}, we present the counterexample showing the insufficiency of the bipartite EoP gap for 4-partite stabilizer states. 
        In Sec.~\ref{sec: generalization}, we generalize the EoP to multipartite systems, and prove the main structure theorem for $2$-producible states.
        In Sec.~\ref{sec: recoverability}, we discuss the operational meaning of generalized EoP in terms of the state redistribution and recoverability. 
        In Sec.~\ref{sec: GSD}, we calculate the generalized EoP for states with general Schmidt decompositions. 
        Finally, we conclude our results and discuss some limitations in Sec.~\ref{sec: conclusion}.

    \section{Preliminaries} \label{sec: preliminaries}

    This section introduces some preliminaries, including entanglement of purification, 2-producible states, and quantum Markov property.
    For more details, see Appendix~\ref{app: preliminaries}.

    \subsection{Entanglement of Purification}

    For a mixed state $\rho_{AB}$, it can be purified into a pure state $\ket{\psi_{ABC}}$ with an ancillary system $C$, such that $\rho_{AB} = \mathrm{Tr}_C \ket{\psi_{ABC}}\!\bra{\psi_{ABC}}$.
    The purification of $\rho_{AB}$ is not unique, and EoP is defined as the minimum of the entanglement entropy between $AC_1$ and $BC_2$ over all the possible purifications and partitions of system $C$ into $C_1$ and $C_2$
    \begin{equation}
        E_p(A:B) \equiv \min_{C_1:C_2} S_{AC_1}(\ket{\psi_{ABC}}).
    \end{equation}
    Consider a given purification of $\rho_{AB}$, such as the canonical purification 
    \begin{equation}
        \ket{\sqrt{\rho_{AB}}} = \rho_{AB}^{1/2} \ket{\Phi_{A\bar{A}}} \otimes \ket{\Phi_{B\bar{B}}},
    \end{equation}
    where $\bar{A}$ and $\bar{B}$ are the replica of $A$ and $B$, and $\ket{\Phi_{X\bar{X}}}$ denotes the unnormalized Bell state between $X$ and $\bar{X}$ for $X = A,B$.
    With the Schmidt decomposition, the purifications of $\rho_{AB}$ can be compressed into $AB\bar{A}\bar{B}$, and thus, are equivalent to the canonical purification of $\rho_{AB}$ up to local isometries to $\bar{A}\bar{B}$.
    Thus, EoP can be expressed as 
    \begin{equation}
        E_p(A:B) = \min_{\hat{U}} S_{AC_1}(\hat{U}_{C\rightarrow\bar{A}\bar{B}} \ket{\sqrt{\rho_{AB}}}),
    \end{equation}
    where $\hat{U}_{C\rightarrow\bar{A}\bar{B}}$ is an isometry from Hilbert space on system $C$ to the space on system $\bar{A}\bar{B}$.
    The lower and upper bounds of the EoP are 
    \begin{equation}
        I(A:B)/2 \leq E_p(A:B) \leq \min\{S_A, S_B\} .
    \end{equation} 
    In this work, we focus on the gap  between $2E_p(A:B)$ and the mutual information $I(A:B)$, 
    \begin{equation}
        g(A:B) = 2 E_p(A:B) - I(A:B),
    \end{equation}
    which serves as a key quantity for characterizing multipartite entanglement.

    \subsection{2-Producible States}

        \begin{figure*}
            \centering
            \includegraphics[width=\textwidth]{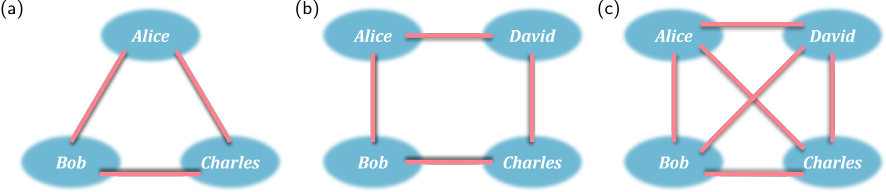}
            \caption{Diagrams of 2-producible states. The pink lines represent the bipartite state between systems.
            (a) Triangle state, the 2-producible state in the tripartite system.
            (b) Polygon state, a special 2-producible state in the $4$-partite system.
            (c) General form of a 2-producible state in the $4$-partite system.}
            \label{figure1}
        \end{figure*}
        We mainly focus on the multipartite pure state.
        In the absence of multipartite entanglement, a pure state is 2-producible, i.e., a tensor product of at most 2-partite states,
        \begin{equation}
            \ket{\psi} = \bigotimes_{i<j} \ket{\psi_{A_{i,j} A_{j,i}}},
        \end{equation}
        where $A_{i} = \otimes_{j \neq i} A_{i,j}$ for $i = 1, 2, \dots, n$, and the state $\ket{\psi_{A_{i,j} A_{j,i}}}$ is a bipartite state between $A_{i,j}$ and $A_{j,i}$.
        For a diagrammatic illustration, see Fig.~\ref{figure1}.

        In a tripartite system, a pure state $\ket{\psi_{ABC}}$ is $2$-producible
        \begin{equation}
            \ket{\psi_{ABC}} = \ket{\psi_{A_1 B_2}} \otimes \ket{\psi_{B_1 \bar{B}}} \otimes \ket{\psi_{A_2\bar{A}}},
        \end{equation}
        where $C = \bar{A}\otimes \bar{B}$, if and only if the EoP gap $g(A:B) = 2 E_p(A:B) - I(A:B)$ vanishes~\cite{PhysRevLett.126.120501}.
        Therefore, the states with a nonvanishing EoP gap $g(A:B)>0$ contain nontrivial multipartite entanglement.
        Naturally, one may expect that the EoP gap $g(A:B)$ can also be used to characterize more than tripartite entanglement in multipartite system.
        However, in Sec.~\ref{sec: 2-producible}, we will show that it does not always hold.

    \subsection{Quantum Markov Property} 

        The quantum Markov property serves as the key technique for proving the structure theorem of the EoP gap in the tripartite case~\cite{PhysRevLett.126.120501}.
        Generalizing this theorem to multipartite systems is based on the framework of quantum Markov recovery.
        Furthermore, the approximate quantum Markov recovery clarifies that the EoP gap quantifies the local recoverability of a tripartite quantum state from its marginals.

        A state $\rho_{ABC}$ has the quantum Markov property, if there exists a unital completely positive trace-preserving (CPTP) map $\mathcal{R}_{B\rightarrow BC}$ such that 
        \begin{align} \label{eq: QMC}
            \rho_{ABC} = \mathcal{I}\otimes \mathcal{R}_{B\rightarrow BC}(\rho_{AB}).
        \end{align}
        This condition is equivalent to the vanishing of the conditional mutual information (CMI)~\cite{hayden2004structure}
        \begin{equation}
            I(A:C|B) \equiv S_{AB} + S_{BC} - S_{ABC} - S_B = 0,
        \end{equation}
        and the state is in the form
        \begin{equation}
            \rho_{ABC} = \sum_l p_l \rho_{AB_1}^l \otimes \rho_{B_2 C}^l,
        \end{equation}
        where $B = B_1 \otimes B_2$.
        This state forms a short quantum Markov chain $A \rightarrow B \rightarrow C$~\cite{accardi1983markovian,hayden2004structure}, and the map $\mathcal{R}_{B\rightarrow BC}$ is called the Markov recovery map.

        States with $I(A:C|B) \neq 0$ deviate from the quantum Markov chain.
        The CMI characterizes the bound on how well the state can be recovered from its marginal~\cite{fawzi2015quantum,PhysRevLett.115.050501,sutter2016universal,sutter2018approximate}
        \begin{align} \label{eq: Markov}
            I(A:C|B) &\geq D_{\mathbb{M}}(\rho_{ABC}| \mathcal{R}_{B\rightarrow BC}(\rho_{AB})) \\
            &\geq -2\log F(\rho_{ABC}, \mathcal{R}_{B\rightarrow BC}(\rho_{AB})), \nonumber
        \end{align}
        where $\mathcal{R}_{B\rightarrow BC}$ is a CPTP map, $D_{\mathbb{M}}(\cdot|\cdot)$ is the measured relative entropy~\cite{donald1986relative,PhysRevLett.103.160504,sutter2018approximate}, and $F(\rho | \sigma) = \Vert \sqrt{\rho} \sqrt{\sigma} \Vert_1$ is the fidelity between $\rho$ and $\sigma$.

        The EoP gap can be expressed as a sum of conditional mutual information involving the canonical purification as
        \begin{align}
            g(A:B) & = I(\bar{A}:B\bar{B}|A) + I(\bar{B}: A|B) \\
            & = I(\bar{B}:A\bar{A}|B) + I(\bar{A}: B|A). \nonumber
        \end{align}
        The vanishing of $g(A:B)$ implies that the canonical purification forms a short quantum Markov chain $\bar{A} \rightarrow A \rightarrow B \rightarrow \bar{B}$.
        For $g(A:B) > 0$, with Petz recovery map [see Eq.~(\ref{eq: Petz}) in Appendix~\ref{app: preliminaries}], we have
        \begin{equation} \label{eq: QMC_g}
            g(A:B) \geq D_{\mathbb{M}}(\rho_{ABC}\Vert \mathcal{R}_{A|B \rightarrow ABC}^{\mathrm{LOCC}}(\rho_{AB})),
        \end{equation}
        where $\mathcal{R}_{A|B \rightarrow ABC}^{\mathrm{LOCC}}$ is a local operation and classical communication (LOCC) recovery channel.

    \section{Entanglement of Purification is Insufficient to Witness More Than Tripartite Entanglement} \label{sec: 2-producible}
    
    In Ref.~\cite{PhysRevLett.126.120501}, a structure theorem is proved that a tripartite pure state is 2-producible if and only if 
    \begin{equation} \label{eq: g}
        g(A:B) \equiv 2 E_{P}(A:B) - I(A:B)
    \end{equation}
    vanishes.
    Therefore, $g(A:B)$ can be viewed as a witness of tripartite entanglement in the tripartite system.
    One may expect to use it to witness the multipartite entanglement in the multipartite system.
    Obviously, in an $n$-partition system $(A_1,A_2,\dots,A_{n})$, $g(A_i:A_j)>0$ implies the existence of tripartite entanglement between $A_i, A_j$ and their complementary $\bigotimes_{l\neq i,j}^n A_l$.
    On the other hand, for an $n$-partite $2$-producible state, it is obvious that $g(A_i:A_j)=0$ for all possible pairs $(A_i:A_j)$. 
    We thus wonder whether $g(A_i:A_j)$ vanishing for all possible pairs $(A_i:A_j)$ of system implies that the state is a 2-producible state, abbreviated as $\mathrm{BPS}$.
    This question is equivalent to whether $g(A_i:A_j) = 0$ for all possible $A_{j(\neq i)}$ implies that $A_i$ is entangled with all possible $A_{j(\neq i)}$ only in the presence of bipartite entanglement, which is proved in Appendix~\ref{app: lemma}.
    However, the answer is negative.
    In the following, we illustrate a counter-example.
    The following proposition is required, which is proved in Appendix~\ref{app: lemma}.
    \begin{proposition} \label{proposition: ME}
        For an $n$-partition $(A_1,A_2,\dots,A_{n})$ of the system, denote the event that the system $A_i$ is entangled with its complementary system $\bigotimes_{j\neq i} A_j$ as $\mathrm{ENT}_i$, the event that the bipartite marginal state of $A_i$ and $A_j$ is entangled as $\mathrm{BE}_{i,j}$, and the event that the subsystem $A_i$ has multipartite entanglement with other subsystems as $\mathrm{ME}_i$, then the probability of the events $\mathbb{P}(\mathrm{ENT}_i)$, $\mathbb{P}(\mathrm{BE}_{i,j})$, and $\mathbb{P}(\mathrm{ME}_{i})$ satisfy
        \begin{equation}
            \mathbb{P}(\mathrm{ENT}_i) \leq \sum_{j\neq i} \mathbb{P}(\mathrm{BE}_{i,j}) + \mathbb{P}(\mathrm{ME}_{i}).
        \end{equation}
    \end{proposition}\noindent
    Note that in the notation of this proposition, the equivalence between two statements phrased in above can be state as 
    \begin{align}
      & \{g(A_i:A_j) = 0, \ \forall A_i, A_j\} = \mathrm{BPS} \nonumber \\
      & \Leftrightarrow \{g(A_i:A_j) = 0, \ \forall A_{j(\neq i)}\} = \overline{\mathrm{ME}}_i,
    \end{align}
    where $\overline{\mathrm{ME}}_i$ is the opposite event of $\mathrm{ME}_i$.

    Here, we consider the random stabilizer state as an example.
    A stabilizer state is the $+1$ eigenstate of the stabilizer operators, which form an Abelian subgroup of Pauli group. 
    Two stabilizer states are related to each other by Clifford operators, which is the stabilizer subgroup of Pauli group in unitary group $SU(2^N)$ and forms the $3$-design of $SU(2^N)$~\cite{zhu2016,PhysRevA.96.062336}.
    It is proved to be locally equivalent to Bell states and GHZ states in tripartite systems~\cite{10.1063/1.2203431}
    \begin{align}
        \ket{\psi_{ABC}^{\mathrm{st}}} & = \mathcal{U}_A\mathcal{U}_B \mathcal{U}_C \ket{0}^{\otimes s_A}  \ket{0}^{\otimes s_B}  \ket{0}^{\otimes s_C}  \ket{\mathrm{GHZ}}^{\otimes g_{ABC}} \nonumber\\
        &~~~~ \otimes \ket{\mathrm{EPR}}^{\otimes e_{AB}} \ket{\mathrm{EPR}}^{\otimes e_{BC} } \ket{\mathrm{EPR}}^{\otimes e_{AC}},
    \end{align}
    where $e_{AB}$, $e_{BC}$, and $e_{AC}$ are the number of Bell pairs between systems $AB$, $BC$, and $AC$, respectively, and $g_{ABC}$ is the number of GHZ states between systems $A$, $B$, and $C$, and $\mathcal{U}_X$ for $X = A,B,C$ are local Clifford operators on system $X$.
    Therefore, the entanglement entropy is 
    \begin{equation}
        S(A) = (e_{AB} + e_{AC} + g_{ABC}) \log 2.
    \end{equation}
    For the stabilizer state, $g(A:B)$ measures the number of GHZ states~\cite{nguyen2018entanglement,PRXQuantum.2.030313}
    \begin{equation}
        g(A:B) = g(B:C) = g(A:C) = g_{ABC} \log 2,
    \end{equation}
    and the logarithmic negativity measures the number of Bell states~\cite{10.1063/1.2203431}
    \begin{equation}
        E_N(A:B) = e_{AB} \log 2.
    \end{equation}
    The random stabilizer state is a uniformly distributed ensemble of stabilizer states.
    Since Clifford group is the $3$-design of $SU(2^n)$, the random stabilizer state is the $3$-design of the Haar random state.
    Thus, the average entanglement entropy of the stabilizer state is~\cite{PhysRevA.74.062314}
    \begin{equation}
        \bar{S}(A) \geq \log D_{A\min} - \frac{D_{A\min}}{D_{A\max}},
    \end{equation}
    where $D_{A\min} = \min \{D_A,D_{BC}\}$, $D_{A\max} = \max \{D_A,D_{BC}\}$, and $D_{X}$ for $X = A,B,C,\dots$ denotes the dimension of the system $X$.
    In qubit systems, $D_X = 2^{N_X}$ for $X = A, B, C,\dots$, where $N_X$ is the number of qubits in system $X$, and we denote the total number of qubits as $N$.
    With the definition of entropy, $\bar{S}(A) \leq \log D_{A\min}$, so the stabilizer state almost fulfills the maximal entanglement.
    After some calculations, we have
    \begin{equation}
        \bar{e}_{AB} \log 2 = \frac{1}{2}(\bar{S}_A + \bar{S}_B - \bar{S}_C - \bar{g}(A:B)) \geq 0.
    \end{equation}
    In particular, for $n_C \equiv N_C/N > 1/2$, it follows that $\bar{g}(A:B) \leq D_{AB}/D_C$ and $\bar{e}_{AB} \leq D_{AB}/D_C$.

    We focus on the $4$-partite random stabilizer state $\ket{\psi_{ABCD}^{\mathrm{st}}}$ with the proportions of the system sizes $N_A:N_B:N_C:N_D = 1:3:3:3$.
    Consider the tripartition $(A:B:CD)$, since the tripartite stabilizer state is equivalent to the tensor product of GHZ states and Bell states, the marginal state $\rho_{AB}$ of the system $AB$ is entangled, i.e., $\mathrm{BE}_{A,B}$, if and only if $E_N(A:B) = e_{AB} \log 2>0$. 
    Similarly, we have
    \begin{equation}
        \mathbb{P}(\mathrm{BE}_{A_i,A_j}) = \mathbb{P}(E_N(A_i:A_j)>0),
    \end{equation}
    for $A_i,A_j \in \{A,B,C,D\}$.
    Assume that the system $A_i$ has at most bipartite entanglement with other systems $A_j$ if $g(A_i:A_j) = 0$ for all $A_j \neq A_i$, i.e.,
    \begin{equation}
        \overline{\mathrm{ME}}_i = \{g(A_i:A_j) = 0, \ \forall A_{j(\neq i)}\} .
    \end{equation}
    Therefore, we have 
    \begin{align}
        \mathbb{P}(\mathrm{ME}_A) & = \mathbb{P}\left(\bigcup_{X \neq A} \{g(A:X)>0\}\right) \nonumber \\
        & \leq \sum_{X \neq A} \mathbb{P}(g(A:X)>0).
    \end{align}
    It follows that 
    \begin{equation} \label{eq: ENT}
        \mathbb{P}(\mathrm{ENT}_{A}) \leq \sum_{X \neq A} [\mathbb{P}(E_N(A:X)>0) + \mathbb{P}(g(A:X)>0)].
    \end{equation}

    Note that $g(A:B)$ and $E_N(A:B)$ are only defined for the marginal state $\rho_{AB}$, so for the $3$-partition $(A:B:CD)$, we have 
    \begin{align}
        E_N(A:B) & = e_{AB} \log 2 , \\
        g(A:B) & = g_{AB(CD)}\log 2 ,  
    \end{align}
    where $e_{AB}$ is the number of Bell pair between $A$ and $B$, and $g_{AB(CD)}$ is the number of GHZ states in the $3$-partition $(A:B:CD)$. 
    Since $e_{AB}$ and $g_{AB(CD)}$ are integer, we have 
    \begin{align}
        & \mathbb{P}(E_N(A:B)>0) = \mathbb{P}(e_{AB}\geq 1) \nonumber \\
        & \leq \sum_{e} e \mathbb{P}(e_{AB} = e)  = \bar{e}_{AB} \leq \frac{D_{AB}}{D_{CD}}, \\
        & \mathbb{P}(g(A:B)>0) = \mathbb{P}(g_{AB(CD)}\geq 1)  \nonumber \\
        & \leq \sum_{g} g \mathbb{P}(g_{AB(CD)} = g) = \bar{g}_{AB(CD)} \leq \frac{D_{AB}}{D_{CD}}, 
    \end{align}
    and similarly 
    \begin{align}
        \mathbb{P}(E_N(A:C)>0) \leq \frac{D_{AC}}{D_{BD}}, \\
        \mathbb{P}(g(A:C)>0) \leq \frac{D_{AC}}{D_{BD}}, \\
        \mathbb{P}(E_N(A:D)>0) \leq \frac{D_{AD}}{D_{BC}}, \\
        \mathbb{P}(g(A:D)>0) \leq \frac{D_{AD}}{D_{BC}}.
    \end{align}
    With the fixed proportions $N_A:N_B:N_C:N_D = 1:3:3:3$, in thermodynamic limit $N \rightarrow \infty$, we have 
    \begin{equation}
        \mathbb{P}(E_N(A:B)>0) \rightarrow 0, \quad \mathbb{P}(g(A:B)>0) \rightarrow 0.
    \end{equation}
    With Eq.~(\ref{eq: ENT}), it follows that
    \begin{equation}
        \mathbb{P}(\mathrm{ENT}_{A}) \rightarrow 0.
    \end{equation}

    On the contrary, note that the system $A$ is entangled with its complementary system $BCD$, i.e., $\mathrm{ENT}_A$, if and only if $S(A)>0$,
    \begin{equation}
        \mathbb{P}(\mathrm{ENT}_{A}) = \mathbb{P}(S(A)>0).
    \end{equation}
    Since $\bar{S}(A) \geq N_A \log 2 - D_{A}/D_{BCD}$, with the typicality of random stabilizer state~\cite{PhysRevA.74.062314}
    \begin{equation}
        \mathbb{P}[S(A)< \bar{S}(A) - \delta] \leq \exp\left(-\frac{\delta^2}{64 N}\right),
    \end{equation}
    we have for an arbitrary small $\epsilon \rightarrow 0^+$ that
    \begin{align}
        \mathbb{P}(S(A)>0) & = 1 - \mathbb{P}(S(A)\leq 0) \geq 1 - \mathbb{P}(S(A)<\epsilon) \nonumber \\
        & \geq 1 - \exp\left(-\frac{[\bar{S}(A)-\epsilon]^2}{64 N}\right) \nonumber \\
        & \rightarrow 1 - \exp\left(- \frac{n_A^2 \log^2 2}{64} N\right) \rightarrow 1,
    \end{align}
    where $n_A \equiv N_A/N = 0.1$.
    It follows that for an arbitrary small $\epsilon > 0$, there exist a sufficient large $N$, such that $\mathbb{P}(\mathrm{ENT}_{A})< \epsilon < 1 - \epsilon <\mathbb{P}(\mathrm{ENT}_{A})$.
    This contradiction shows that $g(A_i:A_j) = 0$ for all possible pairs $(A_i:A_j)$ is not sufficient for the 2-producible state.
    For generic $4$-partite states in $1:3:3:3$ partition, similar analysis shows that bipartite EoP gaps typically fail to characterize $2$-producible states for the generic $4$-partite pure states.
    See Appendix~\ref{app: failure} for more details.

    \section{Generalization of Entanglement of Purification and 2-Producible State} \label{sec: generalization}
    
    \begin{figure*}
        \centering
        \includegraphics[width=\textwidth]{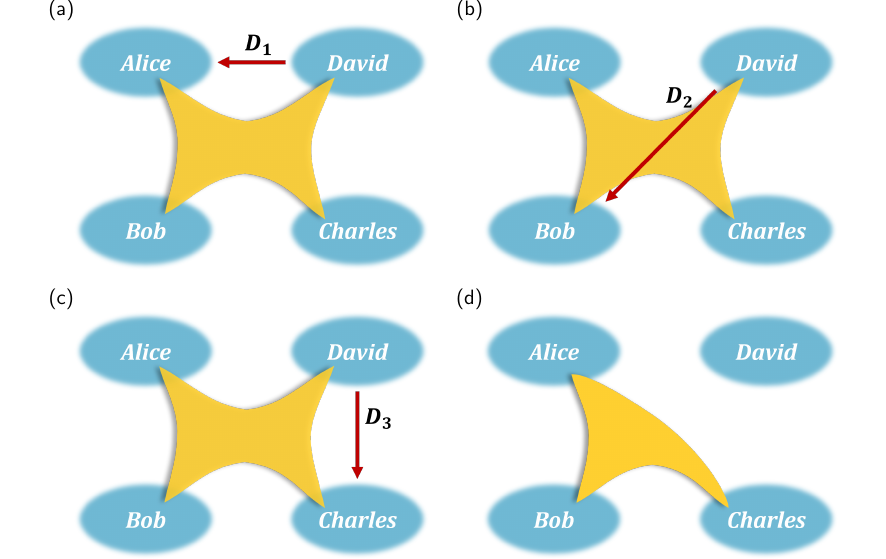}
        \caption{Diagram of state redistribution in $4$-partite system.
        The yellow regions represent multipartite entanglement between systems.
        (a) David sends one part $D_1$ of its system $D$ to Alice, which costs quantum communications $I(D_1: BC|A)/2$.
        (b) Then, David sends another part $D_2$ of its system $D$ to Bob, which costs quantum communications $I(D_2: ACD_1|B)/2$.
        (c) Finally, David sends the last part $D_3$ of its system $D$ to Charles, which costs quantum communications $I(D_3: ABD_1D_2|C)/2$.
        (d) After the redistribution, the total state is only shared by Alice, Bob, and Charles. The optimal total quantum communication cost of the redistribution is $g(A,B,C)$.
        }
        \label{figure2}
    \end{figure*}
    To obtain the sufficient condition for the 2-producible state, we generalize EoP to the multipartite system.
    \begin{definition}
        For an $n$-partite system $(A_1,\dots,A_n)$, $n$-partite generalized entanglement of purification of the state $\rho_{\bigotimes_{i=1}^n A_i}$ is defined as
        \begin{align}
            & E_p(A_1\dots,A_n) = \frac{1}{2} \min \sum_{i=1}^n S(A_i R_i) \nonumber \\
            & = \frac{1}{2} \min D\left(\rho_{\bigotimes_{i=1}^{n+1} A_i} \Big\Vert \bigotimes_{i=1}^n \rho_{A_i R_i}\right),
        \end{align}
        where the minimization is over all possible purifications $\rho_{\bigotimes_{i=1}^{n+1} A_i} = \ket{\psi_{\bigotimes_{i=1}^{n+1} A_i}}\!\bra{\psi_{\bigotimes_{i=1}^{n+1} A_i}}$ of the state $\rho_{\bigotimes_{i=1}^n A_i}$, and $A_{n+1} = \bigotimes_{i=1}^n R_i$ is the ancillary system for purification.
    \end{definition} 
    In particular, for $n=2$, our definition reduces to EoP, $E_p(A,B) = E_p(A:B)$.
    Since the relative entropy is monotonic under quantum channels, considering the discarding of the ancillary system $A_{n+1}$, we have 
    \begin{align} \label{eq: lower_bound}
        & 2 E_p(A_1,\dots,A_n) = \min D\left(\rho_{\bigotimes_{i=1}^{n+1} A_i} \Big\Vert \bigotimes_{i=1}^n \rho_{A_i R_i}\right)  \\
        & \geq D\left(\rho_{\bigotimes_{i=1}^n A_i} \Big\Vert \bigotimes_{i=1}^n \rho_{A_i}\right)= \sum_{i=1}^n S(A_i) - S\left(\bigotimes_{i=1}^n A_i\right). \nonumber
    \end{align}
    Therefore, the $n$-partite generalization of the EoP gap can be defined as
    \begin{align}
        & g(A_1,\dots,A_n) \nonumber \\
        & = \frac{1}{2} \min_{(R_1: \dots,R_n)} \left[ \sum_{i=1}^n [S(A_i R_i) - S(A_i)] + S\left(\bigotimes_{i=1}^n A_i\right) \right]\nonumber \\
        & = \frac{1}{2} \min_{(R_1: \dots,R_n)} \sum_{i=1}^n I\left(R_i: \bar{A}_i \bar{R}_i | A_i\right).
    \end{align}
    Here, $\bar{A}_i = \bigotimes_{j\neq i} A_j$ and $\bar{R}_i = \bigotimes_{j>i} R_j$.
    We introduce the factor $1/2$ in the definition of the generalized EoP gap $g(A_1,\dots,A_n)$, which is different from the existing definition of the EoP gap $g(A:B)$ for the $2$-partite case, i.e., $g(A,B) = g(A:B)/2$.

    With the generalized EoP gap defined above, we have the following structure theorem for the 2-producible state.
    \begin{theorem} \label{theorem: 2-producible}
        The $n$-partite $(A_1,A_2,\dots,A_n)$ pure state $\ket{\psi_{\bigotimes_{i=1}^n A_i}}$ is a $2$-producible state, i.e.,
        \begin{equation}
            \ket{\psi_{\bigotimes_{l=1}^{n} A_l}} = \bigotimes_{i=1}^{n} \bigotimes_{j>i} \ket{\psi_{A_{i,j} A_{j,i}}},
        \end{equation}
        with $A_i = \bigotimes_{j \neq i} A_{i,j}$ for all $i = 1, \dots, n$, if there exists a chain of system set $\alpha_2 \subset \cdots \subset \alpha_k \subset \cdots \subset \alpha_{n-1}$, where $\alpha_n = (A_{i_1}, \cdots , A_{i_n})$, such that $g(\alpha_n) = 0$.
        Then, $g(\alpha) = 0$ for all possible system sets $\alpha = (A_{i_1}, A_{i_2}, \dots)$.             
    \end{theorem} \noindent
    This theorem implies that the necessary and sufficient condition of $2$-producible pure state is that all the generalized EoP gaps $g(\alpha)$ vanish.
    To prove this, we require the following lemmas.     
    The proofs of theorem and lemmas are shown in Appendix~\ref{app: 2-producible}.
    \begin{lemma} \label{lemma: monotonicity}
        $E_p(A_1,\dots,A_n)$ never increase when discarding the system
        \begin{equation}
            E_p(A_1,\dots,A_n) \leq E_p(A_1',\dots,A_n') ,
        \end{equation}
        where $A_i' = A_i \otimes E_i'$.
    \end{lemma}

    \begin{lemma} \label{lemma: reduce}
        A $(n+1)$-partite system $(A_1,A_2,\dots,A_n,A_{n+1})$ pure state up to local unitary operations on $A_i$ is in the form
        \begin{equation} 
            \ket{\psi_{\bigotimes_{i=1}^{n+1} A_i}} = \bigotimes_{i = 1}^n \ket{\psi_{A_{i,1} R_i}} \otimes \ket{\psi_{\bigotimes_{i =1}^n A_{i,2}}},
        \end{equation}
        where $A_{n+1} = \bigotimes_{i=1}^n R_i$, and $A_i = A_{i,1} \otimes A_{i,2}$ for $i=1,\dots, n$, if and only if $g(A_1,\dots,A_n) = 0$.
    \end{lemma} \noindent

    \section{Quantum State Redistribution and Local Recoverability} \label{sec: recoverability}

    Note that the conditional mutual information $I(A:C|B)$ of a pure state $\ket{\psi}_{ABCD}$ is twice the optimal quantum communication cost of the state redistribution~\cite{10.1063/1.1459754,PhysRevLett.100.230501,PhysRevLett.115.050501}, in which Alice (holding systems $A$ and $D$) transfers the system $A$ to Bob (holding the system $B$) with a reference system $C$.
    The generalized EoP gap $g(A_1,A_2,\dots,A_n)$ is the optimal total quantum communication cost to redistribute the system $A_{n+1}$ to the other $n$ parts $A_i, \dots, A_n$.

    Specifically, one redistributes the part $R_{n}$ of the system $A_{n+1}$ to the system $A_n$, which costs $I(R_n: \bar{A}_n | A_n)/2$ for quantum communications.
    Then, one redistributes the part $R_{n-1}$ of the system $A_{n+1}$ to the system $A_{n-1}$, which costs $I(R_{n-1}: \bar{A}_{n-1} \bar{R}_{n-1}| A_{n-1})/2$ for quantum communications.
    Since now $\bar{R}_{n-1}= R_n$ has been redistributed to $A_n$, it is a part of the reference system of the redistributions. 
    Repeat this procedure until redistributing the part $R_1$ of the system $A_{n+1}$ to the system $A_1$, where $I\left(R_i: \bar{A}_i \bar{R}_i | A_i\right)/2$ is the cost to redistribute the system $R_i$ to the system $A_i$ after systems $R_j$ have been redistributed to systems $A_j$ for $j > i$.
    Since the generalized EoP gap $g(A_1,\dots,A_n)$ is optimal over all partitions of $A_{n+1}$, the optimal total cost of the above procedure achieves the generalized EoP gap $g(A_1,\dots,A_n)$.
    See Fig.~\ref{figure2} for a diagrammatic illustration.

    Theorem~\ref{theorem: 2-producible} implies that quantum communications are required to redistribute one party of the multipartite entangled state to other parts.
    In particular, Lemma~\ref{lemma: reduce} implies that the way to redistribute subsystem $A_{n+1}$ of the state with $g(A_1,\dots,A_n) = 0$
    \begin{equation} 
        \ket{\psi_{\bigotimes_{i=1}^{n+1} A_i}} = \bigotimes_{i = 1}^n \ket{\psi_{A_{i,1} R_i}} \otimes \ket{\psi_{\bigotimes_{i =1}^n A_{i,2}}},
    \end{equation}
    to the other parts $A_i$ without quantum communications.    
    Discarding the system $A_{n+1}$, the marginal state is 
    \begin{equation}
        \rho_{\bigotimes_{i=1}^n A_i} = \bigotimes_{i = 1}^n \rho_{A_{i,1}} \otimes \ket{\psi_{\bigotimes_{i =1}^n A_{i,2}}}\!\bra{\cdot},
    \end{equation}
    where $\bra{\cdot}$ denote the bra state of the corresponding ket state $\ket{\psi_{\bigotimes_{i =1}^n A_{i,2}}}$ for simplicity. 
    The state $\rho_{A_{i,1}}$ can be purified locally on each system $A_i$ to $A_i R_i$.
    Denote the channel $\mathcal{R}_{A_{i} \rightarrow A_{i} R_i}$ as the purification map, thus the state $\ket{\psi_{\bigotimes_{i=1}^{n+1} A_i}}$ can be redistributed without quantum communication with the local  operation $\bigotimes_{i=1}^{n}\mathcal{R}_{A_{i} \rightarrow A_{i} R_i}$.

    Redistributing states with a nonvanishing generalized EoP gap $g(A_1,\dots,A_n) > 0$ requires quantum communications.
    Therefore, such states cannot be locally recovered exactly.
    With the Petz recovery map, Eq.~(\ref{eq: Petz}), it follows that
    \begin{align}
        & g(A_1, \dots, A_n)  \\
        & \geq D_{\mathbb{M}}(\rho_{\bigotimes_{i=1}^{n+1} A_i} \Vert \mathcal{R}_{\bar{A}_{n+1} \rightarrow A}^{\mathrm{LOCC}}(\rho_{\bigotimes_{i=1}^n A_i})), \nonumber
    \end{align}
    where 
    \begin{align}
        & \mathcal{R}_{\bar{A}_{n+1} \rightarrow A}^{\mathrm{LOCC}} = \int_{-\infty}^{\infty} \mathrm{d}t \beta_0(t) \bigotimes_{i=1}^n \mathcal{R}_{A_i \rightarrow A_i R_i}^{t}, \\
        & \mathcal{R}_{A_i \rightarrow A_i R_i}^{t}(\cdot) = \rho_{A_i R_i}^{\frac{1+\mathrm{i} t}{2}} \rho_{A_i}^{-\frac{1+\mathrm{i} t}{2}} (\cdot) \rho_{A_i}^{-\frac{1-\mathrm{i} t}{2}} \rho_{A_i R_i}^{\frac{1-\mathrm{i} t}{2}},
    \end{align}
    with $\bar{A}_{n+1} = \bigotimes_{i=1}^n A_i$ and $A = \bigotimes_{i=1}^{n+1} A_i$, for simplicity.
    This means that the generalized EoP gap characterizes the local recoverability of multipartite pure states, and, in particular, the local recovery is realized by the convex combination of rotated Petz recovery map.

    By relaxing the quantifier of local recoverability from the measured relative entropy to fidelity, the recovery can be realized by local operations, the Petz recovery map, only. 
    The proof is shown in Appendix~\ref{app: recovery}, which follows the technique of Ref.~\cite{fawzi2015quantum}.
    
    \begin{theorem} \label{theorem: recovery_k}
        For a $(n+1)$-partite pure state $\rho_{A} = \rho_{\bigotimes_{i=1}^{n+1} A_i}$, with an $n$-partite marginal state $\rho_{\bar{A}_{n+1}} = \rho_{\bigotimes_{i=1}^n A_i}$, by using the local Petz recovery map
        \begin{equation}
            \mathcal{R}_{\bar{A}_{n+1} \rightarrow A}^{\mathrm{LO}} = \bigotimes_{i=1}^n \mathcal{R}_{A_i \rightarrow A_i R_i}^{0},
        \end{equation}
        where the partition $A_{n+1} = \bigotimes_{i=1}^n R_i$ attains the minimum of $g(A_1, \dots, A_n)$,
        it has
        \begin{align}
            g(A_1, \dots, A_n) \geq - \log F(\rho_{A}, \mathcal{R}_{\bar{A}_{n+1} \rightarrow A}^{\mathrm{LO}}(\rho_{\bar{A}_{n+1}})).
        \end{align}
    \end{theorem} \noindent
    In this case, this technique implies the approximate recovery via local operations in fidelity, while the techniques of the measured relative entropy~\cite{PhysRevLett.115.050501,sutter2016universal,sutter2018approximate}, i.e., the rotated Petz recovery maps, implies the approximate recovery with the LOCC operations in the measured relative entropy.
    It differs from the case of conditional mutual information, where the result of the measured relative entropy is strictly tighter than that of Ref.~\cite{fawzi2015quantum}.

    Moreover, the generalized EoP gaps also characterize the distance between a state and 2-producible states in the relative entropy.
    \begin{proposition} \label{proposition: k_BPS}
        For an $n$-partite pure state $\rho_{\bigotimes_{i=1}^n A_i}$, it has
        \begin{equation}
            2 \max_{|\alpha|<n} g(\alpha) \leq \min_{\sigma\overset{\mathrm{LU}}{\sim} \rho} \min_{\mu \in \mathrm{BPS}_k} D(\sigma \Vert \mu),
        \end{equation}
        where $\mathrm{BPS}_n$ denote the set of $n$-partite $2$-producible states
        \begin{equation}
            \ket{\psi_{\bigotimes_{l=1}^{n} A_l}} = \bigotimes_{i=1}^{n} \bigotimes_{j>i} \ket{\psi_{A_{i,j} A_{j,i}}},
        \end{equation}
        and $\mu$ is the state equivalent to $\rho_{\bigotimes_{i=1}^n}$ up to local unitary operations.
    \end{proposition}

    \section{Generalized Entanglement of Purification of States Admitting the General Schmidt Decomposition} \label{sec: GSD}

    In this section, we calculate the generalized EoP gap $g(A_1,\dots,A_{n+1})$ for the multipartite states having the general Schmidt decomposition~\cite{PERES199516,PhysRevA.59.3336},
    \begin{equation}
        \ket{\psi} = \sum_l \sqrt{p_l} \bigotimes_{i=1}^{n+1} \ket{\psi_{A_i}^l},
    \end{equation}
    where $\ket{\psi_{A_i}^l}$ forms the orthonormal basis of system $A_i$.
    The GHZ state, with a form 
    \begin{equation}
        \ket{\mathrm{GHZ}} = \frac{1}{\sqrt{2}}(\ket{0}^{\otimes n+1} + \ket{1}^{\otimes n+1}),
    \end{equation}
    is one typical state having the general Schmidt decomposition, which seems to be the maximally entangled state for such type of states~\cite{GISIN19981}.

    To calculate the generalized EoP and its gap, we embed the system d$A_{n+1}$ to a reference system $R$, where $R = \bigotimes_{i=1}^n R_i$ by using the map $\ket{\phi_{A_{n+1}^l}} \rightarrow \bigotimes_{i=1}^{n} \ket{l}_{R_i}$.
    Thus, the pure state is locally equivalent to
    \begin{equation}
        \ket{\psi_{AR}} = \sum_l \sqrt{p_l} \bigotimes_{i=1}^{n} \ket{\psi_{A_i}^l} \otimes \ket{l}_{R_i}.
    \end{equation}
    Assume that the minimum of $E_p(A_1,\dots,A_n)$ is attained at $\mathcal{U}_{R}\ket{\psi_{AR}}$ with the partition $(A_{1}R_1, A_2 R_2, \dots, A_n R_n)$, where $\mathcal{U}_{R}$ is the unitary operation on system $R$.
    The generalized EoP is
    \begin{equation}
        E_p(A_1,\dots,A_n) = \frac{1}{2} \sum_{i=1}^n S(A_i R_i).
    \end{equation}
    By considering the orthogonality of $\ket{\psi_{A_i}^l}$, the marginal state of system $A_i R_i$ is
    \begin{equation}
        \rho_{A_i R_i} = \sum_l p_l \ket{\psi_{A_i}^l}\!\bra{\psi_{A_i}^l} \otimes \mathcal{N}_i^l(\ket{l}_{R_i}\!\bra{l}),
    \end{equation}
    where $\mathcal{N}_i^l(\rho_{R_i}) = \mathrm{Tr}_{\bar{R}_i}[\mathcal{U}_{R}(\rho_{R_i} \otimes \ket{l}_{\bar{R}_i}\!\bra{l})]$, and $\bar{R}_i = \bigotimes_{j\neq i} R_j$.
    Thus, the entropy is
    \begin{equation}
        S(A_i R_i) = \sum_l p_l S(\mathcal{N}_i^l(\ket{l}_{R_i}\!\bra{l})) + H(p_l) \geq H(p_l),
    \end{equation}
    where $H(p_l)$ is the classical entropy of the probability $\{p_l\}$, and the equality holds if the channel $\mathcal{N}_i^l =\mathcal{I}_{R_i}$ is the identical channel, which universally holds for all $\mathcal{N}_i^l$ if the unitary operation $\mathcal{U}_{R} = \mathcal{I}_{R}$ is the identical channel.
    Thus, the generalized EoP gap is
    \begin{equation}
        E_p(A_1,\dots,A_n) = \frac{1}{2} \min \sum_{i=1}^{n} S(A_i R_i) = \frac{n}{2} H(p_l).
    \end{equation}

    To calculate the generalized EoP gap $g(A_1,\dots,A_k)$, we also need to calculate the marginal entropy $S(A_i)$ and $S(R)$.
    It is easy to see that 
    \begin{equation}
        S(A_i) = S(R) = H(p_l),
    \end{equation}
    so the generalized EoP gap is 
    \begin{align}
        g(A_1,\dots,A_n) & = \frac{1}{2} \left[\sum_{i=1}^n S(A_i R_i) - S(A_i) + S(R)\right] \nonumber \\
        & = \frac{1}{2} H(p_l).
    \end{align}
    In particular, $g(A_1,\dots,A_k) = g(A_i,A_j) = g(A_i:A_j)/2$ for all pairs of subsystems $(A_i,A_j)$.
    This implies that the multipartite states having the general Schmidt decomposition is a 2-producible state if and only if the EoP gaps $g(A_i:A_j)$ for any possible pairs $(A_i,A_j)$ vanish.
    
    Moreover, denoting the event that the $4$-partite stabilizer state has the general decomposition as $\mathrm{Sch}$, we have 
    \begin{align}
        \mathbb{P}(\mathrm{ENT}_A) & = \mathbb{P}(\mathrm{ENT}_A \cap \mathrm{Sch}) + \mathbb{P}(\mathrm{ENT}_A \cap \overline{\mathrm{Sch}}) \nonumber \\
        & \leq \mathbb{P}(\mathrm{ENT}_A \cap \mathrm{Sch}) + \mathbb{P}(\overline{\mathrm{Sch}}),
    \end{align}
    where $\overline{\mathrm{Sch}}$ is the opposite event of $\mathrm{Sch}$.
    According to the discussions in Sec.~\ref{sec: 2-producible}, fixing the proportion $N_A : N_B : N_C : N_D = 1:3:3:3$, in thermodynamic limit $N \rightarrow \infty$,
    \begin{align}
        & \mathbb{P}(\mathrm{ENT}_A \cap \mathrm{Sch}) \rightarrow 0, \\
        & \mathbb{P}(\mathrm{ENT}_A) \rightarrow 1,
    \end{align}
    implies
    \begin{equation}
        \mathbb{P}(\overline{\mathrm{Sch}}) \rightarrow 1.
    \end{equation}
    Therefore, in the proportion $N_A : N_B : N_C : N_D = 1:3:3:3$, the $4$-partite stabilizer states typically do not fulfill the general Schmidt decomposition, which is different from the $3$-partite case where the states are locally equivalent to the tensor product of the GHZ states and Bell pairs~\cite{10.1063/1.2203431}.

    \section{Conclusion and Discussion} \label{sec: conclusion}

    In this work, we investigate the $2$-producible states, which are the tensor product of bipartite states.
    For the random stabilizer state, we shows that the EoP gap $g(A:B)$ is not sufficient to characterize the $2$-producible states of more than tripartite systems.
    Therefore, we have introduced the generalizations of EoP and their gap to the lower bounds in multipartite cases, and proved that a multipartite state is 2-producible if and only if all the generalized EoP gaps vanish.
    This provides a quantitative characterization of the multipartite entanglement with 2-producible states.
    The generalized EoP gap also quantifies the optimal quantum communication cost to redistribute one part of the system to other parts, and thus quantifies the local recoverability of the state from its  marginal state of some parts.
    This shows the operational meaning for the generalized EoP gap.
    It also measures the distance of the state from 2-producible states in relative entropy.
    Moreover, we calculate the generalized EoP gap for the multipartite states having general Schmidt decompositions, where the generalized EoP gap equals the EoP gap in the bipartite case (up to a factor from definition).
    Thus, the states are 2-producible if and only if the EoP gap vanishes for all possible pairs of systems, which implies that the $4$-partite stabilizer states do not always have a general Schmidt decomposition.
    
    Furthermore, there are some limitations on our results.
    We will discuss them in the following.
    Nevertheless, multipartite entanglement is inherently complex, and capturing all its finer aspects with a single framework is impossible. 
    We believe our work provides a rigorous criterion for a specific class of multipartite entanglement (2-producibility), but we acknowledge that a finer characterization requires further investigation. 
    We hope that our results will promote further investigations and understanding of the multipartite entanglement.

    \subsection{Limitations on mixed states}
    First, our results only apply to pure states.
    For mixed states, it only provides a necessary condition for the 2-producible states.
    A $2$-producible mixed state is defined as the convex combination of the $2$-producible pure states, i.e.,
    \begin{equation}
        \rho_{\bigotimes_{i=1}^n A_i} = \sum_l p_l \ket{\psi_{\bigotimes_{i=1}^n A_i}^l} \bra{\cdot},
    \end{equation}
    where $\ket{\psi_{\bigotimes_{i=1}^n A_i}^l}$ is a $2$-producible pure state for all $l$, and $\bra{\cdot}$ is the corresponding bra vector.
    An $n$-partite mixed state $\rho_{\bigotimes_{i=1}^n A_i}$ can be purified into a $(n+1)$-partite pure state $\ket{\psi_{\bigotimes_{i=1}^{n+1} A_i}}$.
    Theorem~1 shows that the purification $\ket{\psi_{\bigotimes_{i=1}^{n+1} A_i}}$ is $2$-producible if and only if a decadent chain of generalized EoP gaps vanish.
    If the purification $\ket{\psi_{\bigotimes_{i=1}^{n+1} A_i}}$ is $2$-producible, then the mixed state $\rho_{\bigotimes_{i=1}^n A_i}$ is also $2$-producible,  
    Thus, vanishing generalized EoP gaps for a mixed state imply it is 2-producible.

    However, due to the bounds for convex combinations $\rho = \sum_l p_l \rho^l$, Eq.~(\ref{eq: convex}), a 2-producible mixed states, where $g(\alpha_k)_{\rho^l} = 0$, may have non-vanishing generalized EoP gaps. 
    For example, the tripartite marginal state of a sum of polygon state ~\cite{PhysRevLett.126.120501} 
    \begin{equation}
        \ket{\psi_{ABCD}} = \sum_{l} p_l \ket{\psi_{A_1 B_2}^l} \otimes \ket{\psi_{B_1 C_2}^l} \otimes \ket{\psi_{C_1 D_2}^l} \otimes \ket{\psi_{D_1 A_2}^l},
    \end{equation}
    is a 2-producible mixed state.
    For the partition $(A|B|CD)$, a sum of polygon state is a sum of triangle state~\cite{PhysRevLett.126.120501}, it follows that $g(A:B) = H(p_l)$~\cite{7fhc-hv9x}, and similarly $g(B:C) = g(A:C) = H(p_l)$.
    Therefore, Theorem~\ref{theorem: 2-producible} cannot be trivially  extended to the mixed states.

    Moreover, the purifications of $2$-producible mixed state have structured multipartite entanglement, such as the SOPS states.
    Characterizing the $2$-producible mixed states thus essentially requires the understanding of their structured multipartite entanglement.
    Note that the Markov gaps vanishes $h(A:B) = h(B:C) = h(A:C)= 0$ for the marginal state $\rho_{ABC}$ of the sum of polygon state, it is interesting to explore whether the multipartite generalization of Markov gap can characterize the $2$-producible mixed states, which requires further investigations.

    \subsection{Limitations on numerical calculation}

    \begin{figure*}
        \centering
        \includegraphics[width=\textwidth]{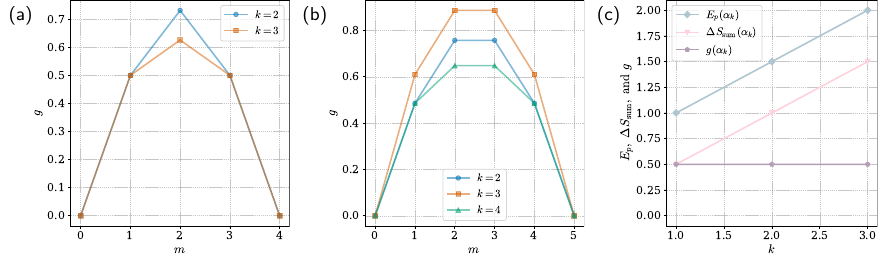}
        \caption{Numerics of generalized EoP gaps for the Dicke states and GHZ states.
        (a) Generalized EoP gap $g(\alpha_k)$ for the Dicke state $\ket{D_{n,m}}$ with $n = 4$, where $\alpha_k = (A_1,A_2,\dots,A_k)$ and $2\leq k \leq n-1$. 
        (b) Generalized EoP gap $g(\alpha_k)$ for the Dicke state $\ket{D_{n,m}}$ with $n = 5$, where $\alpha_k = (A_1,A_2,\dots,A_k)$ and $2\leq k \leq n-1$.
        (c) Generalized EoPs $E_p(\alpha_k)$, their lower bounds $\Delta S_{\mathrm{sum}}(\alpha_k) = \sum_{A_i \in \alpha_k} S(A_i) - S(\bigotimes_{A_i\in \alpha_k} A_i)$, and their gaps $g(\alpha_k)$ for the GHZ state with $n=5$, where $\alpha_k = (A_1,A_2,\dots,A_k)$ and $2\leq k \leq n-1$.
        Since the Dicke states and GHZ states are permutation invariant, the generalized EoP gaps with a fixed subsystem length $|\alpha_k| = k$ are the same.
        }
        \label{figure3}
    \end{figure*}
    Second, the calculation of the generalized EoP and its gap involves an optimization over all possible purifications and partitions, which is an intractable task theoretically.
    In numerics, the minimization can be performed with gradient descent over the unitary group acting on the purifying system $A_{n+1}$.
    We calculate the generalized EoP gap for the Dicke states
    \begin{equation}
        \ket{D_{n,m}} = \sum_{\hat{P}_n \in S_n} \hat{P}_n \ket{0}^{\otimes(n-m)} \ket{1}^{\otimes m}/\sqrt{C_n^m},
    \end{equation}
    and the GHZ states.
    The results for Dicke states with $n = 4,5$ are shown in Fig.~\ref{figure3}(a) and~(b), and the results for GHZ states with $n= 5$ are shown in Fig.~\ref{figure3}(c).
    For details of the calculation, see Appendix~\ref{app: numerical}.

    However, the gradient descent method may converge to a local minimum instead of the global minimum, and it is difficult to verify the global optimality of the solution.
    Therefore, we illustrate some properties of the generalized EoP and its gaps, which may be helpful to exclude the false numerical results.
    \begin{enumerate}
      \item[(1)] Upper bound:
      \begin{align}
          E_p(A_1,\dots, A_k) \leq \frac{1}{2} \left[\sum_{i=1}^{k-1} S(A_i) + S\left(\bigotimes_{i=1}^{k-1} A_i\right)\right],
      \end{align}
      where the equality holds when the state $\rho_{\bigotimes_{i=1}^k A_i}$ is pure.
      \item[(2)] Lower bound:
      \begin{align}
          & E_p(A_1,\dots, A_k A_k') \geq \frac{1}{2} \left[\sum_{i=1}^{k} S(A_i) + S\left(\bigotimes_{i=1}^{k-1} A_i\right) \right. \nonumber \\
          &~~~~ \left.- S\left(\bigotimes_{i=1}^{k-1} A_i \Big|A_k'\right) - S\left(\bigotimes_{i=1}^{k-1} A_i \Big|A_k'\right)\right].
      \end{align}
      \item[(3)] For a descending chain of subsystem $\alpha_k \subset \alpha_{k+1}$,
      \begin{equation}
          E_p(\alpha_{k}) \leq E_p(\alpha_{k+1}).
      \end{equation}
      \item[(4)] Subadditivity:
      \begin{align}
          & E_p(A_1, \dots, A_k) + E_p(A_1',\dots,A_k') \nonumber \\
          & \geq E_p(A_1A_1',\dots,A_kA_k').
      \end{align}
      \item[(5)] Polygamy: For pure state $\rho_{\bigotimes_{i=1}^{k-1} A_i A_k A_k'}$
      \begin{align}
          & E_p(A_1,\dots,A_{k-1},A_k) + E_p(A_1,\dots,A_{k-1},A_k') \nonumber \\
          & \geq E_p(A_1,\dots,A_{k-1},A_kA_k').
      \end{align}
      \item[(6)] For the convex combination $\rho = \sum_l p_l \rho^l$  
      \begin{align}
          & \sum_l p_l E_p(\alpha_k)_{\rho^l} \leq E_p(\alpha_k)_{\rho} \nonumber \\
          & \leq \sum_l p_l E_p(\alpha_k)_{\rho^l} + \frac{k}{2} H(p_l), \\
          & \sum_l p_l g(\alpha_k)_{\rho^l} - \frac{k}{2} H(p_l) \leq g(\alpha_k)_{\rho} \nonumber \\
          & \leq \sum_l p_l g(\alpha_k)_{\rho^l}+ \frac{k+1}{2} H(p_l). \label{eq: convex}
      \end{align}
    \end{enumerate}
    For details to prove the properties, see Appendix~\ref{app: properties}.

    \subsection{Limitations on general $k$-producible states}

    There are other probes that can be used to detect the $k$-producibility of states, for example, the quantum Fisher information.
    For $k$-producible state $\rho_{k-\mathrm{prod}}$, its Fisher information 
    \begin{equation}
      F[\rho_{k-\mathrm{prod}}] \leq (h_{\max} -h_{\min})^2 (s k^2 + r^2),
    \end{equation}
    where $s = \lfloor N/k \rfloor$, $r = N - sk$, $N$ is the total number of particles, and $h_{\max}$ ($h_{\min}$) is the maximum (minimum) eigenvalue of the local Hamiltonian $\hat{H}_i$ acting on each particle~\cite{PhysRevA.85.022321}.
    The quantum Fisher information also implies that the $(k+1)$-producible states are more sensitive than the $k$-producible states for phase estimation.
    However, it assumes that the local Hilbert space have the same dimension for each particle and the local Hamiltonian is the same for each particle, thus it cannot characterize the $k$-producible states completely. 

    In Ref.~\cite{balasubramanian2014multiboundary}, a characterization of $k$-producible states based on entanglement-generic is established.
    It derives a lower bound on the difference 
    \begin{equation}
        \Delta S_{\mathrm{sum}} = \sum_{A_i \in \alpha} S(A_i) - S\left(\bigotimes_{A_i \in \alpha} A_i\right),
    \end{equation}
    which is the lower bound of generalized EoP $E_p(\alpha)$, in terms of the length $|\alpha|$ for $k$-producible pure state $\rho_{\bigotimes_{i=1}^n A_i}$,
    \begin{equation}
        \mathbb{E}_{|\alpha| = j}\Delta S_{\mathrm{sum}} \geq \delta(j,k,n) \sum_{i = 1}^n S(A_i),
    \end{equation} 
    where $\mathbb{E}_{|\alpha| = j}$ averages over $\alpha \subseteq \{A_1, A_2, \dots, A_n\}$ with $|\alpha| = j$, and $\delta(j,k,n)$, called as the fractional entropy deficit, is a positive constant depending on $j,k,n$.
    
    The above two probes provide necessary conditions for the $k$-producible states.
    In Theorem~1, provide a necessary and sufficient condition for the 2-producibility in pure states, which captures more precise information than the witnesses like the quantum Fisher information and the entanglement-generic characterizations.
    However, as a trade-off, our framework specifically in the $k=2$ case cannot be extended to the $k$-producible states with $k>2$.
    The necessary and sufficient characterization of the $k$-producible states with $k>2$ is beyond the scope of this paper, which requires further investigations.

    \begin{acknowledgments}
        This work was supported by the National Natural Science Foundation of China (Grants No.~T2121001, No.~92265207, No.~92365301, No.~12475017), the Natural Science Foundation of Guangdong Province (Grant No.~2024A1515010398), the Guangdong Provincial Quantum Science Strategic Initiative (Grant No.~GDZX2505004), and the Scientific Research Innovation Capability Support Project for Young Faculty (Grant No.~SRICSPYF-ZY2025171). 
    \end{acknowledgments}

    \bibliography{two_producible_state.bib}

    \include{two_producible_state_SI.tex}

\end{document}

%% file: two_producible_state_SI.tex
\appendix

\section{Additional preliminaries} \label{app: preliminaries}

    This Appendix supplies more preliminaries.

    \subsection{Entanglement of Purification}

    Asymptotically, EoP is equal to the entanglement cost 
    \begin{align}
        & E_{\mathrm{LOq}}(\rho_{AB})  \\
        & = \lim_{\epsilon\rightarrow 0} \inf \left\{\frac{m}{n}: d(\Lambda_{\mathrm{LOq}}(\ket{\phi}\!\bra{\phi}^{\otimes m}), \rho_{AB}^{\otimes n}) \leq \epsilon \right\}, \nonumber
    \end{align}
    of creating a state $\rho_{AB}$ from the Bell states $\ket{\phi}$ with local operations with negligible communications (LOq)~\cite{10.1063/1.1498001}
    \begin{equation}
        E_{\mathrm{LOq}}(\rho_{AB}) = \lim_{n \to \infty} \frac{1}{n} E_p(\rho_{AB}^{\otimes n}) \equiv E_p^{\infty}(\rho_{AB}).
    \end{equation}

    Moreover, it also satisfies several properties~\cite{PhysRevA.91.042323}, such as the monotonicity under discarding of a quantum system
    \begin{equation}
        E_p(A:BC) \geq E_p(A:B),
    \end{equation}
    the polygamy for a tripartite pure state $\ket{\psi_{ABC}}$
    \begin{equation}
        E_p(A:B) + E_p(A:C) \geq E_p(A:BC),
    \end{equation}
    and the sub-additivity on the tensor product
    \begin{equation}
        E_p(A_1 A_2 : B_1 B_2) \leq E_p(A_1 : B_1) + E_p(A_2 : B_2).
    \end{equation}

    In the perspective of holographic duality, EoP is related to the entanglement wedge cross-section, which is the minimal cross-section of the entanglement wedge of a bipartite system and is suggested to measure the bipartite correlations for a given mixed state~\cite{umemoto2018entanglement}.
    A related measure, called reflected entropy $S_R$, is also proposed, which is defined as the entanglement entropy for the canonical purification of $\rho_{AB}$ between $A\bar{A}$ and $B\bar{B}$~\cite{dutta2021canonical,hayden2021markov,akers2022reflected}
    \begin{equation}
        S_R(A:B) \equiv S_{A\bar{A}}(\ket{\sqrt{\rho_{AB}}}) .
    \end{equation}
    The reflected entropy is bounded by
    \begin{equation}
        I(A:B) \leq S_R (A:B) \leq 2\min\{S_A, S_B\},
    \end{equation}
    and is conjectured as $S_R = 2 E_p$ for the holographic state.
    For a general quantum state, they are quite different, e.g., the reflected entropy is not monotonic under discarding of a quantum system~\cite{PhysRevA.107.L050401}.
    
    The gap between the reflected entropy and mutual information,
    \begin{equation}
        h(A:B) = S_R(A:B) - I(A:B)
    \end{equation}
    is called the Markov gap, which is suggested to measure tripartite entanglement in holographic states~\cite{akers2020entanglement}.
    For a tripartite pure state $\ket{\psi_{ABC}}$, if the Markov gap vanishes, i.e., $h(A:B)=0$, the state is a sum of triangle state (SOTS)~\cite{PhysRevLett.126.120501}
    \begin{equation}
        \ket{\psi_{ABC}} = \sum_l \sqrt{p_l} \ket{\psi_{A_1 B_2}^l} \otimes \ket{\psi_{B_1 C_2}^l} \otimes \ket{\psi_{C_1 A_2}^l},
    \end{equation}
    where $\mathcal{H}_X = \bigoplus_l (\mathcal{H}_{X_1}^l \otimes \mathcal{H}_{X_2}^l) \oplus \mathcal{H}_X^0$ for $X = A, B, C$, and $\ket{\psi_{X_1 Y_2}^l} \in \mathcal{H}_{X_1}^l \mathcal{H}_{Y_2}^l$.

    \subsection{$k$-producible states}

        In an $n$-partite system $(A_1, A_2, \dots, A_n)$, a pure state $\ket{\psi}$ is producible by $k$-partite entanglement ($k$-producible, for short) if it can be written as the tensor product of at most $k$-partite states, i.e.~\cite{guhne2005multipartite,balasubramanian2014multiboundary}
        \begin{equation}
            \ket{\psi} =  \bigotimes_{i=1}^{m} \ket{\psi_i},
        \end{equation}
        where $\ket{\psi_i}$ is a state located on at most $k$ parts of the system and $m \geq n/k$ is the number of components of tensor product.
        A state is genuine $k$-partite entangled state, if it is not $(k-1)$-producible~\cite{guhne2005multipartite}.

    \subsection{Quantum Markov chain}

        For $I(A:C|B) \neq 0$, the approximate quantum Markov property~\cite{fawzi2015quantum,PhysRevLett.115.050501,sutter2016universal,sutter2018approximate}
        \begin{align} 
            I(A:C|B) &\geq D_{\mathbb{M}}(\rho_{ABC}| \mathcal{R}_{B\rightarrow BC}(\rho_{AB})) \\
            &\geq -2\log F(\rho_{ABC}, \mathcal{R}_{B\rightarrow BC}(\rho_{AB})), \nonumber
        \end{align}
        is a strength of the strong sub-additivity of entropy $I(A:C|B) \geq 0$.
        Here, $D_{\mathbb{M}}(\cdot|\cdot)$ is the measured relative entropy~\cite{donald1986relative,PhysRevLett.103.160504,sutter2018approximate}, which is defined as the supremum of the relative entropy with measured states over all the possible projective measurements
        \begin{equation}
            D_{\mathbb{M}}(\rho|\sigma) = \sup \left\{D(\mathcal{M}(\rho)|\mathcal{M}(\sigma)): \mathcal{M}(\cdot) = \sum_i \hat{\Pi}_i (\cdot)\hat{\Pi}_i\right\}.
        \end{equation}
        The channel $\mathcal{R}_{B\rightarrow BC}$ is called the approximate Markov recovery channel, which is universal in the sense that it does not depend on the state $\rho_A$ of the system $A$~\cite{sutter2016universal,sutter2018approximate}.
        In particular, the equality $I(A:C|B) = 0$ implies that $F(\rho_{ABC}, \mathcal{R}_{B\rightarrow BC}(\rho_{AB})) = 1$, which re-establishes the exact quantum Markov chain, Eq.~(\ref{eq: QMC}).   

        Furthermore, with Stinespring dilation theorem~\cite{stinespring1955positive}, the approximate quantum Markov recovery can be extended to approximate universal recovery maps of quantum relative entropy~\cite{junge2018universal,sutter2018approximate}, such that
        \begin{equation} \label{eq: Petz}
            D(\rho\Vert \sigma) - D(\mathcal{N}(\rho)\Vert\mathcal{N}(\sigma)) \geq D_{\mathbb{M}}(\rho \Vert \mathcal{R}\circ \mathcal{N}(\sigma)),
        \end{equation}
        where the recovery map $\mathcal{R} = \int \mathrm{d}t \beta_0(t) \mathcal{R}_{\sigma,\mathcal{N}}^t$ is the convex combination of the rotated Petz recovery maps~\cite{petz1986sufficient,petz1988sufficiency}
        \begin{equation} 
            \mathcal{R}_{\sigma,\mathcal{N}}^t(\rho) = \sigma^{\frac{1+\mathrm{i}t}{2}} \mathcal{N}^{\dagger}[\mathcal{N}(\sigma)^{-\frac{1+\mathrm{i}t}{2}} \rho \mathcal{N}(\sigma)^{\frac{1-\mathrm{i}t}{2}}] \sigma^{\frac{1-\mathrm{i}t}{2}}.
        \end{equation}
        Here, $\beta_0(t) = \frac{\pi}{2(\cosh(\pi t) +1)}$ is a probability distribution over the parameter $t$.
        
        The Markov gap,
        \begin{equation}
            h(A:B) = I(\bar{A}:B|A) = I(\bar{B}:A|B),
        \end{equation}
        quantifies the Markov recovery of $\rho_{AB\bar{A}}$ or $\rho_{AB\bar{B}}$ from $\rho_{AB}$~\cite{hayden2021markov}.
        However, in this procedure, the total canonical purification $\rho_{AB\bar{A}\bar{B}}$ cannot be recovered from $\rho_{AB}$, due to the global entanglement from the coherent sum in $\{\sqrt{p_l}\}$.
        With the Petz recovery map, we have the approximate quantum Markov chain for $g(A:B)$ as 
        \begin{equation}
            g(A:B) \geq D_{\mathbb{M}}(\rho_{ABC}\Vert \mathcal{R}_{A|B \rightarrow ABC}^{\mathrm{LOCC}}(\rho_{AB})),
        \end{equation}
        where $C = \bar{A} \otimes \bar{B}$,
        \begin{align}
            \mathcal{R}_{A|B \rightarrow ABC}^{\mathrm{LOCC}} & = \int_{-\infty}^{\infty} \mathrm{d}t \beta_0(t) \mathcal{R}_{A\rightarrow A \bar{A}}^t\otimes \mathcal{R}_{B\rightarrow B \bar{B}}^t, \nonumber \\
            \mathcal{R}_{A\rightarrow A \bar{A}}^t(\cdot) & = \rho_{A\bar{A}}^{\frac{1+\mathrm{i} t}{2}} \rho_{A}^{-\frac{1+\mathrm{i} t}{2}} (\cdot)\rho_{A}^{-\frac{1-\mathrm{i} t}{2}} \rho_{A\bar{A}}^{\frac{1-\mathrm{i} t}{2}}, \\
            \mathcal{R}_{B\rightarrow B \bar{B}}^t(\cdot) & = \rho_{B\bar{B}}^{\frac{1+\mathrm{i} t}{2}} \rho_{B}^{-\frac{1+\mathrm{i} t}{2}} (\cdot)\rho_{B}^{-\frac{1-\mathrm{i} t}{2}} \rho_{B\bar{B}}^{\frac{1-\mathrm{i} t}{2}}.
        \end{align}
        In particular, when the equality in Eq.~(\ref{eq: QMC_g}) holds, each of the rotated Petz maps can recover the state $\rho_{ABC}$, where the maps are local operations, 
        \begin{equation}
            \mathcal{R}_{A|B \rightarrow ABC}^{\mathrm{LO}} = \mathcal{R}_{A\rightarrow A \bar{A}}^t\otimes \mathcal{R}_{B\rightarrow B \bar{B}}^t.
        \end{equation}

\section{Failure of pairwise EoP gap characterizing multipartite entanglement in Haar random states} \label{app: failure}

Here, we consider the failure probability of pairwise EoP gap characterizing multipartite entanglement of a generic $4$-partite pure state $\ket{\psi_{ABCD}}$ in Haar measure.
We still assume the partition of the system $A,B,C,D$ is $n_A:n_B:n_C:n_D = 1:3:3:3$.
In a tripartite Haar random state $\ket{\psi_{ABC}}$, the marginal state $\rho_{AB}$ exhibits a phase transition from separable to entangled state when the proportion $n_C = N_C/N$ of the system $C$ crosses the threshold $n_C = n_{\mathrm{SEP}} = [1+\min(n_A,n_B)]/2$~\cite{}, i.e., $\lim_{N \rightarrow \infty} \mathbb{P}(\mathrm{BE}_{A,B}) = 1$ if $n_C < n_{\mathrm{SEP}}$, and $\lim_{N \rightarrow \infty} \mathbb{P}(\mathrm{BE}_{A,B}) = 0$ otherwise.
In our case, we have $\lim_{N \to \infty} \mathbb{P}_{BE_{A,X}} = 0$ for $X = B,C,D$, since if $n_{\mathrm{SEP}} = 0.55 < n_{\overline{AX}} = 0.6$, where $\overline{AX}$ is the complement of the system $AX$.
Thus, the probability of the system $A$ being entangled with all the other subsystems $B,C,D$ satisfies 
\begin{equation}
    \lim_{N \to \infty} \mathbb{P}(\mathrm{ENT}_A) = \lim_{N \to \infty} \mathbb{P}(\mathrm{ME}_A).
\end{equation}

Denote the event that the pairwise EoP gaps detect the multipartite entanglement of the state as $\mathrm{EoP}$, i.e.,
\begin{align}
    \mathrm{EoP} & = \Big\{\{g(A_i:A_j) = 0, \ \forall A_i, A_j\} = \mathrm{BPS}\Big\} \nonumber \\
    & = \Big\{\forall A_i: \{g(A_i:A_j) = 0, \ \forall A_{j(\neq i)}\} = \overline{\mathrm{ME}}_i\Big\}.
\end{align}
Thus, the probability of the event $\mathrm{ME}_A$ satisfies
\begin{align}
    \mathbb{P}(\mathrm{ME}_A) & \leq \mathbb{P}(\overline{\mathrm{EoP}}) + \mathbb{P}(\mathrm{ME}_A \cap \mathrm{EoP}) \nonumber \\
    & = \mathbb{P}(\overline{\mathrm{EoP}}) + \mathbb{P}\left(\bigcap_{X \neq A}\{g(A:X) = 0\}\right) \nonumber \\
    & \leq \mathbb{P}(\overline{\mathrm{EoP}}) + \min_{X \neq A} \mathbb{P}(g(A:X) = 0).
\end{align}
Notice that EoP gap vanishes $g(A:X) = 0$ if and only if the state $\ket{\psi_{A|X|\overline{AX}}}$ is triangle state, and the Markov gap vanishes $h(A:X) = 0$ if and only if the state $\ket{\psi_{A|X|\overline{AX}}}$ is SOTS~\cite{PhysRevLett.126.120501}.
Since triangle state is a special case of SOTS, we have the inclusion of events $\{g(A:X) = 0\} \subseteq \{h(A:X) = 0\}$, and $\mathbb{P}(g(A:X) = 0) \leq \mathbb{P}(h(A:X) = 0)$.
By the typicality of the Markov gap in Haar random states~\cite{7fhc-hv9x}, we have $\lim_{N \to \infty} \mathbb{P}(h(A:X) = 0) = 0$ if $1/2<n_{\bar{AX}}<3/4$.
Here, we have $n_{\bar{AX}} = 0.6$, thus the failure probability of the pairwise EoP gap characterizing multipartite entanglement satisfies
\begin{equation}
    \mathbb{P}(\overline{\mathrm{EoP}}) \geq \lim_{N \to \infty} \mathbb{P}(\mathrm{ENT}_A).
\end{equation}
Since $\mathrm{ENT}_A = \{S(A) > 0\}$, by page's formula 
\begin{equation}
    S(A) = 0.1 N \log 2 + O(e^{-cN}),
\end{equation}
and typicality of entropy in Haar random states~\cite{}, we have $\lim_{N \to \infty} \mathbb{P}(\mathrm{ENT}_A) = 1$.
In conclusion, the failure probability
\begin{equation}
    \lim_{N \to \infty} \mathbb{P}(\overline{\mathrm{EoP}}) = 1,
\end{equation}
and typically, the pairwise EoP gap fails to characterize multipartite entanglement of a generic $4$-partite pure state $\ket{\psi_{ABCD}}$.

\section{Proof of propositions in Sec.~III} \label{app: lemma}

    \begin{appthm}[proposition]{proposition: ME}
        For a $n$-partition $(A_1,A_2,\dots,A_{n})$ of the system, denote the event that the system $A_i$ is entangled with its complementary system $\bigotimes_{j\neq i} A_j$ as $\mathrm{ENT}_i$, the event that the bipartite marginal state of $A_i$ and $A_j$ is entangled as $\mathrm{BE}_{i,j}$, and the event that the subsystem $A_i$ has multipartite entanglement with other subsystems as $\mathrm{ME}_i$, then the probability of the events $\mathbb{P}(\mathrm{ENT}_i)$, $\mathbb{P}(\mathrm{BE}_{i,j})$, and $\mathbb{P}(\mathrm{ME}_{i})$ satisfy
        \begin{equation}
            \mathbb{P}(\mathrm{ENT}_i) \leq \sum_{j\neq i} \mathbb{P}(\mathrm{BE}_{i,j}) + \mathbb{P}(\mathrm{ME}_{i}).
        \end{equation}
    \end{appthm}

    \begin{proof}
        If the system $A_i$ has no multipartite entanglement, it is entangled with other systems $A_j$ only in bipartite entanglement, i.e.
        \begin{equation} 
            \ket{\psi_{\bigotimes_{l=1}^{n} A_l}} = \bigotimes_{j \neq i} \ket{\psi_{A_{i,j} A_{j,1}}} \otimes \ket{\psi_{\bigotimes_{j\neq i} A_{j,2}}}.
        \end{equation}
        Thus, $A_i$ is entangled with it complementary system $\bigotimes_{j\neq i} A_j$ if and only if there is some system $A_{j\neq i}$ such that the marginal state of $A_iA_j$ is entangled, i.e.
        \begin{equation}
            \mathrm{ENT}_i \cap \overline{\mathrm{ME}}_i = \bigcup_{j\neq i} \mathrm{BE}_{i,j} \cap \overline{\mathrm{ME}}_i.
        \end{equation}
        It follows that
        \begin{align}
            \mathbb{P}(\mathrm{ENT}_i) & = \mathbb{P}(\mathrm{ENT}_i\cap\overline{\mathrm{ME}}_i) +  \mathbb{P}(\mathrm{ENT}_i\cap\mathrm{ME}_i) \nonumber\\
            & = \mathbb{P}\left(\bigcup_{j\neq i} \mathrm{BE}_{i,j}\cap\overline{\mathrm{ME}}_i\right) +  \mathbb{P}(\mathrm{ENT}_i\cap\mathrm{ME}_i)\nonumber\\
            & \leq \sum_{j \neq i} \mathbb{P}(\mathrm{BE}_{i,j}) + \mathbb{P}(\mathrm{ME}_i).
        \end{align}
    \end{proof}

        \begin{proposition}
        For a $n$-partition system $(A_1,A_2,\dots,A_{n})$, the following statements are equivalent: \\
        (1) If $g(A_i: A_j) = 0$ for all $A_{j\neq i}$, $A_i$ has only bipartite entanglement with other $A_j$, i.e.,
        \begin{equation} 
            \ket{\psi_{\bigotimes_{l=1}^{n} A_l}} = \bigotimes_{j \neq i} \ket{\psi_{A_{i,j} A_{j,1}}} \otimes \ket{\psi_{\bigotimes_{j\neq i} A_{j,2}}},
        \end{equation}
        where $A_i = \bigotimes_{j \neq i} A_{i,j}$, $A_j = A_{j,1} \otimes A_{j,2}$ for $j \neq i$, and $A_i$ is entangled with $A_j$ only via bipartite entanglement $\ket{\psi_{A_{i,j} A_{j,1}}}$. \\
        (2) If $g(A_i: A_j) = 0$ for all possible pairs $(A_i: A_j)$, the state $\ket{\psi}$ is a 2-producible state, i.e.,
        \begin{equation}
            \ket{\psi_{\bigotimes_{l=1}^{n} A_l}} = \bigotimes_{i=1}^{n} \bigotimes_{j>i} \ket{\psi_{A_{i,j} A_{j,i}}},
        \end{equation}
        where $A_i = \bigotimes_{j \neq i} A_{i,j}$ for all $i = 1, \dots, n$.
    \end{proposition}\noindent
    \begin{proof}
        (2) $\Rightarrow$ (1) is obvious. \\
        (1) $\Rightarrow$ (2): We prove the proposition by induction on $n$.
        For $n = 3$, it is proved by the structure theorem of $g(A:B)$ in Ref.~\cite{PhysRevLett.126.120501}.
        For $n \geq 3$, applying the assumption to $A_{n}$, it is sufficient to prove that the $(n-1)$-partitie state $\ket{\psi_{\bigotimes_{j\neq n} A_{j,2}}}$ is a 2-producible state.
        Since the mutual information is invariant
        \begin{equation}
            I(A:B) = I(A':B'),
        \end{equation}
        if both $A$ and $B$ are extended into $A' = A \otimes R_A$ and $B' = B \otimes R_B$, where
        \begin{equation}
            \rho_{A'B'} = \rho_{AB} \otimes \rho_{R_A}\otimes \rho_{R_B},
        \end{equation}
        and the entanglement of purification is never increasing upon discarding of quantum system~\cite{PhysRevA.91.042323},
        \begin{equation}
            E_p(A': B') \geq E_p(A : B).
        \end{equation}
        From the definition of $g(A: B)$, Eq.~(\ref{eq: g}), it follows 
        \begin{equation}
            g(A: B) \leq g(A': B').
        \end{equation}
        Here, let $A = A_{i,2}$, $B = A_{j,2}$, $A' = A_i$ and $B' = A_j$,
        \begin{equation}
            0 \leq g(A_{i,2}: A_{j,2}) \leq g(A_i: A_j) = 0,
        \end{equation}
        for all possible pairs $(A_{i,2}: A_{2,j})$.
        The result is proved by the inductive hypothesis. 
    \end{proof}

\section{Proof of Structure Theorem of 2-Producible States} \label{app: 2-producible}

    \begin{appthm}{theorem: 2-producible}
        The $n$-partite $(A_1,A_2,\dots,A_n)$ pure state $\ket{\psi_{\bigotimes_{i=1}^n A_i}}$ is a 2-producible state, i.e.
        \begin{equation}
            \ket{\psi_{\bigotimes_{l=1}^{n} A_l}} = \bigotimes_{i=1}^{n} \bigotimes_{j>i} \ket{\psi_{A_{i,j} A_{j,i}}},
        \end{equation}
        if there exists a chain of system set $\alpha_2 \subset \cdots \subset \alpha_k \subset \cdots \subset \alpha_{n-1}$, where $\alpha_k = (A_{i_1}, \cdots , A_{i_k})$, such that $g(\alpha_k) = 0$,  
        where $A_i = \bigotimes_{j \neq i} A_{i,j}$ for all $i = 1, \dots, n$.
        Then, $g(\alpha) = 0$ for all possible system sets $\alpha = (A_{i_1}, A_{i_2}, \dots)$.             
    \end{appthm}
    To prove this, we require the following lemmas.
    \begin{appthm}[lemma]{lemma: reduce}
        A $(n+1)$-partite system $(A_1,A_2,\dots,A_n,A_{n+1})$ pure state up to local unitary on $A_i$ is in the form
        \begin{equation} 
            \ket{\psi_{\bigotimes_{i=1}^{n+1} A_i}} = \bigotimes_{i = 1}^n \ket{\psi_{A_{i,1} R_i}} \otimes \ket{\psi_{\bigotimes_{i =1}^n A_{i,2}}},
        \end{equation}
        where $A_{n+1} = \bigotimes_{i=1}^n R_i$, and $A_i = A_{i,1} \otimes A_{i,2}$ for $i=1,\dots, n$, if and only if $g(A_1,A_2,\dots,A_n) = 0$.
    \end{appthm}
    \begin{proof}
        Only if part of $g(A_1,A_2,\dots,A_n) = 0$ is obvious by direct calculation.
        If $g(A_1,A_2,\dots,A_n) = 0$, then there exists a decomposition $A_{n+1} = \bigotimes_{i=1}^n R_i$ such that
        \begin{align}
            g(A_1,A_2,\dots,A_n)& \equiv \frac{1}{2} \sum_{i=1}^n [S(A_i R_i) - S(A_i)] + S(\bar{A}_{n+1}) \nonumber \\
            & = \frac{1}{2} \sum_{i=1}^n I\left(R_i: \bar{A}_i \bar{R}_i | A_i\right) = 0,
        \end{align}
        where $\bar{A}_i = \bigotimes_{j\neq i} A_j$ and $\bar{R}_i = \bigotimes_{j>i} R_j$.
        The second equality follows from the definition of conditional mutual information 
        \begin{align}
            & \sum_{i=1}^n I\left(R_i: \bar{A}_i \bar{R}_i | A_i\right) \nonumber \\
            & \equiv \sum_{i=1}^n S(A_iR_i) + S(\bar{A}_{n+1}\bar{R}_i) - S(\bar{A}_{n+1} \bar{R}_{i-1}) - S(A_i) \nonumber \\
            & = \sum_{i=1}^n [S(A_iR_i) - S(A_i)] + S(\bar{A}_{n+1}) - S\left(A\right), 
        \end{align}
        where $A = \bigotimes_{i=1}^{n+1} A_i$ for brief, and $S\left(A\right) = 0$ since the state $\ket{\psi_{A}}$ is pure.
        Then, $g(A_1,A_2,\dots,A_n) = 0$ implies that $I\left(R_i: \bar{A}_i \bar{R}_i | A_i\right) = 0$ for all $i = 1, \dots, n$, since the conditional mutual information is non-negative.
        For $i=1$, we have $I\left(R_1: \bar{A}_1 \bar{R}_1 | A_1\right) = 0$, by the structure theorem of conditional mutual information~\cite{hayden2004structure}, it follows that there exists a decomposition $R_1 = R_{1,1} \otimes R_{1,2}$ such that
        \begin{equation}
            \ket{\psi_{A}} = \ket{\psi_{A_{1,1} R_{1}}} \otimes \ket{\psi_{A_{1,2} (\bigotimes_{j\geq 2} A_j) \bar{R}_1}}.
        \end{equation}
        For $i>1$, we assume the state has the decomposition that
        \begin{equation}
            \ket{\psi_{A}} = \bigotimes_{l\leq i-1} \ket{\psi_{A_{l,1} R_{l}}} \otimes \ket{\psi_{(\bigotimes_{l\leq i-1}A_{l,2}) (\bigotimes_{j\geq i} A_j) \bar{R}_{i-1}}},
        \end{equation}
        where $A_l = A_{l,1} \otimes A_{l,2}$ for $l \leq i-1$.
        By the structure theorem of conditional mutual information $I\left(R_i: \bar{A}_i \bar{R}_i | A_i\right) = 0$ again, there exists a decomposition $R_i = R_{i,1} \otimes R_{i,2}$ such that the marginal state on $\bar{A}_{n+1}\bar{R}_{i-1}$ is
        \begin{align}
            \rho_{\bar{A}_{n+1}\bar{R}_{i-1}} & = \bigotimes_{l\leq i-1} \rho_{A_{l,1}} \otimes \ket{\psi_{(\bigotimes_{l\leq i-1}A_{l,2}) (\bigotimes_{j\geq i} A_j) \bar{R}_{i-1}}}\bra{\cdot} \nonumber\\
            & = \sum_n p_n \rho_{A_{i,1}R_i}^{(n)} \otimes \rho_{A_{i,2} \bar{A}_i \bar{R}_i}^{(n)},
        \end{align}
        where $\bra{\cdot}$ denote the bra state of the corresponding ket state for simplicity. 
        Discarding the system $\bigotimes_{l\leq i-1} A_{l,1}$, we have
        \begin{align}
            &\ket{\psi_{(\bigotimes_{l\leq i-1}A_{l,2}) (\bigotimes_{j\geq i} A_j) \bar{R}_{i-1}}}\bra{\cdot} \nonumber\\
            & = \sum_n p_n \rho_{A_{i,1}R_i}^{(n)} \otimes \rho_{(\bigotimes_{l\leq i}A_{l,2}) (\bigotimes_{j\geq i+1} A_j) \bar{R}_i}^{(n)} \nonumber \\
            & = \ket{\psi_{A_{i,1}R_i}}\bra{\cdot} \otimes \ket{\psi_{(\bigotimes_{l\leq i}A_{l,2}) (\bigotimes_{j\geq i+1} A_j) \bar{R}_i}}\bra{\cdot},
        \end{align}
        and the total state is 
        \begin{equation}
            \ket{\psi_{A}} = \bigotimes_{l\leq i} \ket{\psi_{A_{l,1} R_{l}}} \otimes \ket{\psi_{(\bigotimes_{l\leq i}A_{l,2}) (\bigotimes_{j\geq i+1} A_j) \bar{R}_{i}}}.
        \end{equation}
        The induction stop when $i = n$, where $(\bigotimes_{j\geq i+1} A_j) \bar{R}_{i} = 0$ and the state is
        \begin{equation}
            \ket{\psi_{A}} = \bigotimes_{l = 1}^n \ket{\psi_{A_{l,1} R_{l}}} \otimes \ket{\psi_{\bigotimes_{l=1}^kA_{l,2}}}.
        \end{equation}
    \end{proof}

    \begin{appthm}[lemma]{lemma: monotonicity}
        $E_p(A_1,\dots,A_n)$ never increase when discarding the system
        \begin{equation}
            E_p(A_1,\dots,A_n) \leq E_p(A_1',\dots,A_n') ,
        \end{equation}
        where $A_i' = A_i \otimes E_i'$.
    \end{appthm}
    \begin{proof}
        Let $\ket{\psi_{\bigotimes_{i=1}^n A_i' R_i}}$ be the purification of $\rho_{\bigotimes_{i=1}^n A_i'}$ where $E_p(A_1',\dots,A_n')$ takes the minimum
        \begin{equation}
            E_p(A_1',\dots,A_n') = \frac{1}{2} \sum_{i=1}^{n} S(A_i'R_i) = \frac{1}{2} \sum_{i=1}^{n} S(A_i E_i R_i)
        \end{equation}
        Notice that the state $\ket{\psi_{\bigotimes_{i=1}^n A_i' R_i}}$ is also the purification of $\rho_{\bigotimes_{i=1}^n A_i}$ since $A_i' = A_i \otimes E_i$.
        Now denote $R_i' = R_i \otimes E_i$, we have 
        \begin{align}
            & E_p(A_1',\dots,A_n') = \frac{1}{2} \sum_{i=1}^{n} S(A_i R_i') \nonumber \\
            & \geq E_p(A_1,\dots,A_n) \equiv \frac{1}{2} \min \sum_{i=1}^{n} S(A_i R_i'). 
        \end{align}
    \end{proof}

    \begin{proof}[Proof of Theorem~\ref{theorem: 2-producible}]
        We prove the theorem by induction on the number of partition $n$.
        Without loss of generality, we assume $g(\alpha_{k}) = g(A_1,A_2,\dots,A_{k}) = 0$ for $2 \leq k \leq n-1$.
        For $n = 3$, the result is followed from the structure theorem of $g(A:B)$.
        For $n \geq 3$, apply Lemma~\ref{lemma: reduce} to $g(\alpha_{n-1}) = g(A_1,A_2,\dots,A_{n-1}) = 0$, there exist decomposition $A_{n}= \otimes_{i=1}^{n-1} A_{n,i}$ and $A_i = A_{i,1} \otimes A_{i,2}$ for $i \leq n-1$, such that  
        \begin{equation}
            \ket{\psi_{\bigotimes_{i=1}^n A_i}}  = \bigotimes_{i = 1}^{n-1} \ket{\psi_{A_{i,1} A_{n,i}}} \otimes \ket{\psi_{\bigotimes_{i=1}^{n-1}A_{i,2}}}.
        \end{equation}
        It is sufficient to prove that the $(n-1)$-partite pure state $\ket{\psi_{\bigotimes_{i=1}^{n-1}A_{i,2}}}$ is a 2-producible state.
        Now, for any $g(\alpha_{k}) = g(A_{1}, \dots, A_{k}) = 0$, where $k \leq n-2$, with the tensor-product structure of the state $\ket{\psi_{\bigotimes_{i=1}^{n} A_i}}$,
        \begin{equation}
            S\left(\bigotimes_{i=1}^{k}A_{i}\right) - \sum_{i=1}^{k} S(A_i) 
            = S\left(\bigotimes_{i=1}^{k}A_{i,2}\right) - \sum_{i=1}^{k} S(A_{i,2}),
        \end{equation}
        and by the monotonicity of $E_p(A_1,\dots,A_{k})$ under discarding system
        \begin{equation}
            E_p(A_{1,2},\dots,A_{k,2}) \leq E_p(A_1,\dots,A_{k}),
        \end{equation}
        we have
        \begin{equation}
            0\leq g(A_{1,2}, \dots, A_{k,2}) \leq g(A_{1}, \dots, A_{k}) = 0.
        \end{equation}
        Thus, we have a chain of system set $\alpha_2' \subset \cdots \subset \alpha_k' \subset \cdots \subset \alpha_{n-2}'$, where $\alpha_k' = (A_{1,2}, \cdots , A_{k,2})$, such that $g(\alpha_k') = 0$ for the $(n-1)$-partite pure state $\ket{\psi_{\bigotimes_{i=1}^{n-1}A_{i,2}}}$.
        By the inductive assumption, $\ket{\psi_{\bigotimes_{i=1}^n A_i}}$ is a 2-producible state.
        Moreover, $g(\alpha) = 0$ for all possible system sets $\alpha = (A_{i_1}, A_{i_2}, \dots)$ is directly followed by applying Lemma~\ref{lemma: reduce} to 2-producible state $\ket{\psi_{\bigotimes_{l=1}^{n} A_l}}$.

    \end{proof}

\section{Proof of Approximate Local Recoverability} \label{app: recovery}

    \begin{appthm}{theorem: recovery_k}
        For a $(n+1)$-partite pure state $\rho_{A} = \rho_{\bigotimes_{i=1}^{n+1} A_i}$, with a $n$-partite marginal state $\rho_{\bar{A}_{n+1}} = \rho_{\bigotimes_{i=1}^n A_i}$, by using the local Petz recovery map
        \begin{equation}
            \mathcal{R}_{\bar{A}_{n+1} \rightarrow A}^{\mathrm{LO}} = \bigotimes_{i=1}^n \mathcal{R}_{A_i \rightarrow A_i R_i}^{0},
        \end{equation}
        where the partition $A_{n+1} = \bigotimes_{i=1}^n R_i$ attains the minimum of $g(A_1, \dots, A_n)$,
        it has
        \begin{align}
            g(A_1, \dots, A_n) \geq - \log F(\rho_{A}, \mathcal{R}_{\bar{A}_{n+1} \rightarrow A}^{\mathrm{LO}}(\rho_{\bar{A}_{n+1}})).
        \end{align}
    \end{appthm} \noindent
    We first review some useful lemmas used in the proof.
    \begin{lemma}[Ref.~\cite{fawzi2015quantum}] \label{lemma: relative_entropy}
        For a density matrix $\rho$, a non-negative operator $\sigma$, and a sequence of completely positive trace non-increasing maps $\mathcal{W}_N$ such that 
        \begin{equation}
            \liminf_{N \to \infty} e^{\xi N} \mathrm{Tr}(\mathcal{W}_N(\rho^{\otimes N})) > 0,
        \end{equation} 
        for any $\xi > 0$, then
        \begin{equation}
            \limsup_{N \to \infty} \frac{1}{N} D(\mathcal{W}_N(\rho^{\otimes N}) \Vert \mathcal{W}_N(\sigma^{\otimes N})) = D(\rho \Vert \sigma) .
        \end{equation}
    \end{lemma}

    For a density matrix $\rho$ whose eigenvalues are $p_1 \geq p_2 \geq \cdots \geq p_d$, the eigenvalues of its $N$-times copy $\rho^{\otimes N}$ are $r_{\lambda} = \exp\left(-N[H(\tilde{\lambda}) + D(\tilde{\lambda} \vert p)]\right)$, where $\tilde{\lambda} = \lambda/N$, $\lambda = (\lambda_1, \lambda_2, \cdots, \lambda_d)$ such that $|\lambda| = \sum_i \lambda_i = N$.
    \begin{lemma}[Ref.~\cite{fawzi2015quantum}] \label{lemma: typical}
        Consider the projector $\Pi^{\lambda}_N$ on the eigenspace of eigenvalue $r_{\lambda} = \exp\left(-N[H(\tilde{\lambda}) + D(\tilde{\lambda} \vert p)]\right)$ of $\rho^{\otimes N}$. 
        For any given $\delta>0$, define the typical projector 
        \begin{equation}
            \Pi^{\delta}_{N} = \sum_{\lambda \in \Gamma_{\delta}} \Pi^{\lambda}_N,
        \end{equation}
        where $\Gamma_{\delta} = \{\lambda: |\log (r_{\lambda})/N + H(p)|\leq \delta\}$, then there exists $\epsilon > 0$, such that
        \begin{equation}
            \mathrm{Tr}(\Pi_{\Gamma}^{N} \rho^{\otimes N}) \geq 1 - e^{-N\epsilon}.
        \end{equation}
    \end{lemma}

    \begin{lemma}[Ref.~\cite{fawzi2015quantum}] \label{lemma: fidelity}
        Let $\rho$ and $\sigma$ be non-negative operators and $\hat{W}$ an arbitrary operator, then 
        \begin{equation}
            F(\rho, \hat{W} \sigma \hat{W}^{\dagger}) \geq F(\hat{W}^{\dagger}\rho \hat{W}, \sigma).
        \end{equation}
    \end{lemma}

    \begin{proof}[Proof of Theorem~\ref{theorem: recovery_k}]
        Assume the minimum of $g(A_1, \dots, A_n)$ is attained at the partition $A_{n+1} = \bigotimes_{i=1}^n R_i$
        \begin{equation}
            2 g(A_1, \dots, A_n) = \sum_{i=1}^n [S(A_i R_i) - S(A_i)] + S(\bar{A}_{n+1}).
        \end{equation}
        For the pure state $\rho_{A} = \rho_{\bigotimes_{i=1}^n A_i R_i}$, we consider the typical projectors
        \begin{equation}
            \Pi_{A_i^N}^{\delta} = \sum_{\lambda_{A_i} \in \Gamma_{A_i\delta}} \Pi_{A_i}^{\lambda_{A_i}},  \quad
            \Pi_{A_i^N R_i^N}^{\delta} = \sum_{\lambda_{A_iR_i} \in \Gamma_{A_iR_i\delta}} \Pi_{A_iR_i}^{\lambda_{A_iR_i}}, 
        \end{equation}
        where $\Gamma_{X\delta} = \{\lambda_{X}: |\log (r_{\lambda_{X}})/N + H(p_{X})| \leq \delta\}$ for $X = A_i^N, A_i^N R_i^N$, $p_{X}$ is the probability distribution of the eigenstates of the marginal state $\rho_{X}$, and $H(p_{X}) = S(X)$ is the entropy of $\rho_{X}$.
        By Lemma~\ref{lemma: fidelity}, we have
        \begin{equation}
            \mathrm{Tr}(\Pi_{X^N}^{\delta} \rho_{X}) \geq 1 - e^{-N\epsilon}.
        \end{equation}
        Consider the states $\Gamma_{A}^{(N)} = \mathcal{W}_N(\rho_{A}^{\otimes N})$ and $\mu_{A}^{(N)} = \mathcal{W}_N(\rho_{\bar{A}_{n+1}}^{\otimes N})$, where 
        \begin{align}
            \mathcal{W}_N(\cdot) & = \bigotimes_{i=1}^n \Pi_{A_i^NR_i^N}^{\delta} \Pi_{A_i^N}^{\delta} (\cdot) \bigotimes_{i=1}^n \Pi_{A_i^N}^{\delta} \Pi_{A_i^NR_i^N}^{\delta}.
        \end{align}

        By Lemma~9.4.2 in Ref.~\cite{wilde2013quantum}, we have 
        \begin{align}
            \mathrm{Tr}(\Gamma_{A}^{(n)}) & \geq \mathrm{Tr}\left(\bigotimes_{i=1}^n \Pi_{A_i^NR_i^N}^{\delta} \rho_{A}^{\otimes N}\right)  \nonumber \\
            & -2 \sqrt{1 - \mathrm{Tr}\left(\bigotimes_{i=1}^n \Pi_{A_i^N}^{\delta} \rho_{AB\bar{A}\bar{B}}^{\otimes N}\right)}.
        \end{align}
        Moreover, we have 
        \begin{align}
            &\mathrm{Tr}\left(\bigotimes_{i=1}^n \Pi_{A_i^N}^{\delta} \rho_{A}^{\otimes N}\right) \geq \mathrm{Tr}\left(\bigotimes_{i=1}^{n-1} \Pi_{A_i^N}^{\delta} \rho_{A}^{\otimes N}\right) \nonumber \\ 
            &~~~~- \mathrm{Tr}\left(\bigotimes_{i=1}^{n-1} \Pi_{A_i^N}^{\delta} (1-\Pi_{A_k^N}^{\delta}) \rho_{A}^{\otimes N}\right) ,
        \end{align}
        where for some $\epsilon_k > 0$,
        \begin{align}
            &\mathrm{Tr}\left(\bigotimes_{i=1}^{n-1} \Pi_{A_i^N}^{\delta} (1-\Pi_{A_k^N}^{\delta}) \rho_{A}^{\otimes N}\right) \nonumber \\
            & \leq \mathrm{Tr}\left((1-\Pi_{A_k^N}^{\delta}) \rho_{A}^{\otimes N}\right) \leq e^{-N\epsilon_k}.
        \end{align}
        Iterating the above inequality, we have
        \begin{align}
            \mathrm{Tr}\left(\bigotimes_{i=1}^n \Pi_{A_i^N}^{\delta} \rho_{A}^{\otimes N}\right) & \geq \mathrm{Tr}\left(\Pi_{A_1^N}^{\delta} \rho_{A}^{\otimes N}\right) - \sum_{i=2}^{n} e^{-N\epsilon_i}  \nonumber \\
            & \geq 1 - \sum_{i=1}^{n} e^{-N\epsilon_i} \geq 1 - n e^{-N\epsilon_{\min}} \nonumber \\
            & \geq 1 - e^{-N\epsilon'}, 
        \end{align}
        where $\epsilon_{\min} = \min_i \epsilon_i > 0$ and there exists $0< \epsilon' < \epsilon_{\min} - \frac{1}{N} \log n$ for sufficiently large $N$.
        Similarly, we have $\mathrm{Tr}(\Pi_{A^N \bar{A}^N}^{\delta} \Pi_{B^N\bar{B}^N}^{\delta} \rho_{AB\bar{A}\bar{B}}^{\otimes N}) \geq 1 - e^{- N \epsilon''}$ for sufficiently large $N$, and consequently
        \begin{equation}
            \mathrm{Tr}(\Gamma_{A}^{(n)}) \geq 1 - e^{- N \epsilon'''},
        \end{equation}
        for some $\epsilon'''>0$ and sufficiently large $N$.
        Then, by Lemma~\ref{lemma: relative_entropy}, for arbitrary small $\eta > 0$, there is sufficient large $N$, such that
        \begin{align}
            D(\Gamma_{A}^{(N)} \Vert \mu_{A}^{(N)}) &\leq N[D(\rho_{A} \Vert \rho_{\bar{A}_{n+1}}) + \eta] \nonumber \\
            & \leq N[S(\bar{A}_{n+1}) + \eta].
        \end{align}
        Here, $\bar{A}_{n+1} = \bigotimes_{i=1}^{n} A_i$.
        Since fidelity is the $1/2$-R\'enyi relative entropy, we have
        \begin{equation}
            -2\log F(\Gamma_{A}^{(N)} \Vert \mu_{A}^{(N)}) \leq N[S(\bar{A}_{n+1}) + \eta]. \nonumber \\
        \end{equation}

        For the fidelity, by Lemma~\ref{lemma: fidelity}, we move the projectors $\Pi_{A_i^NR_i^N}^{\delta}$ to the side of $\mu_{AR}^{(N)}$, thus 
        \begin{equation}
            F(\Gamma_{A}^{(N)} \Vert \mu_{A}^{(N)}) = F(\mathcal{V}_N(\rho_{A}^{\otimes N})  \Vert \mu_{A}^{(N)}),
        \end{equation}
        where 
        \begin{equation}
            \mathcal{V}_N(\cdot) = \bigotimes_{i = 1}^n \Pi_{A_i^N}^{\delta} (\cdot)\bigotimes_{i = 1}^n \Pi_{A_i^N}^{\delta}.
        \end{equation}
        For the projectors $\Pi_{X^N}^{\lambda_X}$ into the eigenspace of $\rho_{X}^{\otimes N}$, we have
        \begin{align}
            \Pi_{A_i^N}^{\lambda_{A_i}} & = r_{\lambda_{A_i}}^{1/2} (\rho_{A_i}^{-1/2})^{\otimes N} \Pi_{A_i^N}^{\lambda_{A_i}} \leq  r_{\lambda_{A_i}}^{1/4} (\rho_{A_i}^{-1/2})^{\otimes N}, \\
            \Pi_{A_i^NR_i^N}^{\lambda_{A_iR_i}} & = r_{\lambda_{A_iR_i}}^{-1/2} (\rho_{A_iR_i}^{1/2})^{\otimes N} \Pi_{A_i^NR_i^N}^{\lambda_{A_iR_i}} \leq  r_{\lambda_{A_iR_i}}^{-1/2} (\rho_{A_iR_i}^{1/2})^{\otimes N}, 
        \end{align}
        and therefore
        \begin{align}
            \Pi_{A_i^N}^{\delta} & \leq |\Gamma_{A_i\delta}| e^{-\frac{N}{2}(S(A_i) - \delta)}  (\rho_{A_i}^{-1/2})^{\otimes N}, \\
            \Pi_{A_i^N R_i^N}^{\delta} & \leq |\Gamma_{A_iR_i\delta}| e^{\frac{N}{2}(S(A_i R_i) - \delta)}  (\rho_{A_i R_i}^{1/2})^{\otimes N}, 
        \end{align}
        where $|\Gamma_{X\delta}| \leq \dim \mathrm{Sym}^N(X) \leq (N+1)^{d_X}$ for $X = A_i, A_i R_i$ are polynomials $p_N$ in $N$.
        It follows that 
        \begin{align}
            \mu_{A}^{(N)} & \leq p_N e^{\frac{N}{2}\sum_i [S(A_i R_i) - S(A_i)]} \nonumber\\
            &~~~~\times \mathcal{R}_{\bar{A}_{n+1}\rightarrow A}^{\mathrm{LO}}(\rho_{\bar{A}_{n+1}})^{\otimes N} ,
        \end{align}
        where the channels
        \begin{align}
            \mathcal{R}_{\bar{A}_{n+1} \rightarrow A}^{\mathrm{LO}} & = \bigotimes_{i=1}^n \mathcal{R}_{A_i \rightarrow A_i R_i}^{0}, \\
            \mathcal{R}_{A_i \rightarrow A_i R_i}^{0}(\cdot) & = \rho_{A_iR_i}^{1/2} \rho_{A_i}^{-1/2} (\cdot) \rho_{A_i}^{-1/2} \rho_{A_iR_i}^{1/2}, 
        \end{align}
        are trace-preserving, since
        \begin{equation}
            \mathrm{Tr}_{R_i}(\rho_{A_i}^{-1/2} \rho_{A_i R_i}^{1/2}  \rho_{A_i R_i}^{1/2} \rho_{A_i}^{-1/2}) = \mathrm{Tr}_{R_i}(\rho_{A_i}^{-1} \rho_{A_i R_i}) = \hat{I}_{A_i}.
        \end{equation}
        Since the fidelity is $F(\rho, \sigma) = \mathrm{Tr}\sqrt{\sqrt{\rho}\sigma \sqrt{\rho}}$, by operator monotonicity, we have
        \begin{align}
            & F(\Gamma_{A}^{(N)} \Vert \mu_{A}^{(N)}) 
            \leq p_N e^{\frac{N}{2}\sum_i [S(A_i R_i) - S(A_i)]} \nonumber \\
            &\times F(\mathcal{V}_N(\rho_{A}^{\otimes N}) 
            \Vert \mathcal{R}_{\bar{A}_{n+1}\rightarrow A}^{\mathrm{LO}}(\rho_{\bar{A}_{n+1}})^{\otimes N}).
        \end{align}
        Notice that $\Pi_{A_i^N}^{\delta} \leq \hat{I}_{A_i^N}$, we have $\mathcal{V}_N(\rho_{A}^{\otimes N}) \leq \rho_{A}^{\otimes N}$. 
        By operator monotonicity again, we have
        \begin{align}
            & F(\mathcal{V}_N(\rho_{A}^{\otimes N})  \Vert \mathcal{R}_{\bar{A}_{n+1}\rightarrow A }^{\mathrm{LO}}(\rho_{\bar{A}_{n+1}})^{\otimes N}) \nonumber \\
            & \leq F^N(\rho_{A} \Vert \mathcal{R}_{\bar{A}_{n+1}\rightarrow A}^{\mathrm{LO}}(\rho_{\bar{A}_{n+1}})).
        \end{align}
        Thus, we conclude that 
        \begin{align}
            & -\log F(\rho_{A} \Vert \mathcal{R}_{\bar{A}_{n+1}\rightarrow A}^{\mathrm{LO}}(\rho_{\bar{A}_{n+1}})) \nonumber \\
            & \leq \frac{1}{2}\sum_i [S(A_i R_i) - S(A_i)] + \frac{1}{N} \log p_N \nonumber\\
            &~~~~ -\log F(\Gamma_{A}^{(N)} \Vert \mu_{A}^{(N)}) \nonumber\\
            & \leq g(A_1, \dots, A_n) + \frac{1}{N} \log p_N + \eta.
        \end{align}
        Since $\eta$ can be arbitrarily small, in the limit $N\rightarrow \infty$, we have
        \begin{equation}
            g(A_1, \dots, A_n) \geq -\log F(\rho_{A} \Vert \mathcal{R}_{\bar{A}_{n+1}\rightarrow A}^{\mathrm{LO}}(\rho_{\bar{A}_{n+1}})).
        \end{equation} 

    \end{proof}

    We also consider the upper bound of $g(A_1,\dots, A_n)$.
    \begin{appthm}[proposition]{proposition: k_BPS}
        For a $n$-partite pure state $\rho_{\bigotimes_{i=1}^n A_i}$, it has
        \begin{equation}
            2\max_{|\alpha|<n} g(\alpha) \leq \min_{\sigma\overset{\mathrm{LU}}{\sim} \rho} \min_{\mu \in \mathrm{BPS}_n} D(\sigma \Vert \mu),
        \end{equation}
        where $\mathrm{BPS}_n$ denote the set of $n$-partitie $2$-producible states
        \begin{equation}
            \ket{\psi_{\bigotimes_{l=1}^{n} A_l}} = \bigotimes_{i=1}^{n} \bigotimes_{j>i} \ket{\psi_{A_{i,j} A_{j,i}}},
        \end{equation}
        and $\mu$ is the state equivalent to $\rho_{\bigotimes_{i=1}^n}$ up to local unitary operations.
    \end{appthm} \noindent
    To prove this we require the following lemma.
    \begin{proposition} \label{proposition: k_QMC}
        For a $(n+1)$-partite pure state $\rho_{A}  =  \rho_{\bigotimes_{i=1}^{n+1} A_i}$,
        \begin{equation}
            2 g(\alpha) \leq \min_{\sigma\overset{\mathrm{LU}}{\sim} \rho} \min_{\mu \in \mathrm{QMC}(\alpha)} D(\sigma_{A} \Vert \mu_{A}) ,
        \end{equation}
        where $|\alpha| = n$, i.e. $\alpha = (A_{i_1}, \dots A_{i_{n}})$, and $\mathrm{QMC}(\alpha) = \{\rho: g_{\rho}(\alpha) = 0\}$.
    \end{proposition}

    \begin{proof}
        Without loss of generality, we assume $\alpha = (A_1, \dots, A_n)$.
        We also assume the minimum of $g(A_1, \dots, A_n)$ is attained at the partition $A_{n+1} = \bigotimes_{i=1}^n R_i$
        Consider the quantity 
        \begin{align}
            & D(\rho_{A} \Vert \mu_{A}) -  D(\rho_{\bar{A}_{n+1}} \Vert \mu_{\bar{A}_{n+1}})  \nonumber \\
            &~~~~ + \sum_{i=1}^n [D(\rho_{A_i} \Vert \mu_{A_i}) - D(\rho_{A_i R_i} \Vert \mu_{A_i R_i})] \nonumber \\
            & = 2 g(A_1, \dots, A_n) - \mathrm{Tr}(\rho_{A} \log \mu_{A}) + \mathrm{Tr}(\rho_{\bar{A}_{n+1}} \log \mu_{\bar{A}_{n+1}}) \nonumber \\
            &~~~~ + \sum_{i=1}^n [\mathrm{Tr}(\rho_{A_i R_i} \log \mu_{A_i R_i})-\mathrm{Tr}(\rho_{A_i} \log \mu_{A_i})] . 
        \end{align}
        Since $\mu_{A}$ is pure, $\mathrm{Tr}(\rho_{A} \log \mu_{A}) = 0$.
        In addition, by Lemma~\ref{lemma: reduce} $\mu_{A} = \ket{\psi_{A}} \bra{\psi_{A}} \in \mathrm{QMC}(A_1, \dots, A_n)$, where
        \begin{equation}
            \ket{\psi_{A}} = \bigotimes_{i = 1}^n \ket{\psi_{A_{i,1} R_i}} \otimes \ket{\psi_{\bigotimes_{i =1}^n A_{i,2}}}.
        \end{equation}
        Thus, we have
        \begin{align}
            & \mathrm{Tr}(\rho_{A_i R_i} \log \mu_{A_i R_i}) - \mathrm{Tr}(\rho_{A_i} \log \mu_{A_i}) \\
            & = \mathrm{Tr}(\rho_{A_{i,1}R_i} \log \mu_{A_{i,1}R_i}) -  \mathrm{Tr}(\rho_{A_{i,1}} \log \mu_{A_{i,1}}), \nonumber \\
            & \mathrm{Tr}(\rho_{\bar{A}_{n+1}} \log \mu_{\bar{A}_{n+1}}) = \mathrm{Tr}(\rho_{\bigotimes_{i =1}^n A_{i,2}} \log \mu_{\bigotimes_{i =1}^n A_{i,2}}) \nonumber \\
            &~~~~ + \sum_{i=1}^n \mathrm{Tr}(\rho_{A_{i,1}} \log \mu_{A_{i,1}}),  
        \end{align}
        where $\mu_{A_{i,1}R_i} = \ket{\psi_{A_{i,1}R_i}} \bra{\psi_{A_{i,1}R_i}}$, and $\mu_{\bigotimes_{i =1}^n A_{i,2}} = \ket{\psi_{\bigotimes_{i =1}^n A_{i,2}}} \bra{\psi_{\bigotimes_{i =1}^n A_{i,2}}}$ are pure states.
        It follows that 
        \begin{align}
            & \mathrm{Tr}(\rho_{A_i R_i} \log \mu_{A_i R_i}) - \mathrm{Tr}(\rho_{A_i} \log \mu_{A_i}) \nonumber \\
            & = -  \mathrm{Tr}(\rho_{A_{i,1}} \log \mu_{A_{i,1}}),  \\
            & \mathrm{Tr}(\rho_{\bar{A}_{n+1}} \log \mu_{\bar{A}_{n+1}}) = \sum_{i=1}^n \mathrm{Tr}(\rho_{A_{i,1}} \log \mu_{A_{i,1}}).
        \end{align}
        In conclusion, for $\mu_{A}$
        \begin{align}
            & D(\rho_{A} \Vert \mu_{A}) -  D(\rho_{\bar{A}_{n+1}} \Vert \mu_{\bar{A}_{n+1}})  \nonumber \\
            &~~~~ + \sum_{i=1}^n [D(\rho_{A_i} \Vert \mu_{A_i}) - D(\rho_{A_i R_i} \Vert \mu_{A_i R_i})] \nonumber \\
            & = 2 g(A_1, \dots, A_n).
        \end{align}
        By the strong sub-additivity of relative entropy
        \begin{equation}
            D(\rho_{A_i} \Vert \mu_{A_i}) \leq D(\rho_{A_i R_i} \Vert \mu_{A_i R_i}),
        \end{equation}
        and non-negativity of relative entropy
        \begin{equation}
            D(\rho_{\bar{A}_{n+1}} \Vert \mu_{\bar{A}_{n+1}}) \geq 0,
        \end{equation}
        it follows that 
        \begin{equation}
            2 g(A_1, \dots, A_n) \leq D(\rho_{A} \Vert \mu_{A}),
        \end{equation}
        which prove the proposition.
    \end{proof}

    \begin{proof}[Proof of Proposition~\ref{proposition: k_BPS}]
        Now, consider the $n$-partite pure state $\rho_{\bigotimes_{i=1}^n A_n}$.
        For all possible system set $\alpha = (A_{i_1}, A_{i_2},\dots)$ with $|\alpha|<n$, we have
        \begin{equation}
            2 g(\alpha) \leq \min_{\sigma\overset{\mathrm{LU}(\alpha)}{\sim} \rho} \min_{\mu \in \mathrm{QMC}(\alpha)} D(\sigma \Vert \mu) ,
        \end{equation}
        where $\mathrm{LU}(\alpha)$ denote the equivalence of states up to local unitary on the partition $\alpha \cup \left(\bigotimes_{i \not \in \alpha} A_i\right)$.
        Since $\mathrm{LU}$ is the refinement of all the equivalence $\mathrm{LU}(\alpha)$, i.e. the equivalent class $[\rho]_{\mathrm{LU}} = \bigcap_{\alpha} [\rho]_{\mathrm{LU}(\alpha)}$
        Therefore, 
        \begin{equation}
            2 \max_{|\alpha|<n} g(\alpha) \leq \min_{\sigma\overset{\mathrm{LU}}{\sim} \rho} \min_{\mu \in \bigcap_{|\alpha|<n} \mathrm{QMC}(\alpha)} D(\sigma \Vert \mu).
        \end{equation}
        With Theorem~\ref{theorem: 2-producible}, we have that the set $\mathrm{BPS}_n$ of $2$-producible states on $n$-partition system is the intersection of the sets $\mathrm{QMC}(\alpha)$ of state with $g(\alpha) = 0$, for all possible system set $\alpha= (A_{i_1}, A_{i_2},\dots)$, i.e.
        \begin{equation}
            \mathrm{BPS}_n = \bigcap_{|\alpha|<n} \mathrm{QMC}(\alpha),
        \end{equation}
        and the second inequality is followed.
    \end{proof}


\section{Numerical Method to Compute the Generalized Entanglement of Purification} \label{app: numerical}

    The generalized entanglement of purification 
    \begin{equation}
        E_{p}(A_1,\dots,A_n) = \frac{1}{2} \min \sum_{i=1}^n S(A_i R_i)
    \end{equation}
    requires minimizing the sum of entropies over all possible purifications $\ket{\psi_{\bigotimes_{i=1}^{n+1}A_i}}$ of the state $\rho_{\bigotimes_{i=1}^{n}A_i}$, where $A_{n+1} = \bigotimes_{i=1}^n R_i$.
    Specifically, we start with an initial purification $\ket{\psi_{\bigotimes_{i=1}^{n+1}A_i}}$ of the state $\rho_{\bigotimes_{i=1}^{n}A_i}$.
    We parametrize a unitary $\hat{U}_{A_{n+1}} = \exp (\mathrm{i} \Theta_{A_{n+1}})$ acting on the purifying system $A_{n+1}$.
    
    The differential of the target function $E = \sum_{i=1}^n S(A_i R_i)/2$ with respect to the Hermitian operator $\Theta_{A_{n+1}}$ is 
    \begin{align}
        \mathrm{d} E & = \frac{1}{2} \sum_{i=1}^n \mathrm{d} S(A_i R_i) 
        = -\frac{1}{2} \sum_{i=1}^n \mathrm{d} \mathrm{Tr}_{A_i R_i} \left(\rho_{A_i R_i} \log \rho_{A_i R_i}\right)  \nonumber \\
        & = -\frac{1}{2} \sum_{i=1}^n \mathrm{Tr}_{A_i R_i} \left(\mathrm{d} \rho_{A_i R_i} \log \rho_{A_i R_i}\right) .
    \end{align}
    Here, we use the fact that $\mathrm{d} \log \rho_{A_i R_i} = \rho_{A_i R_i}^{-1} \mathrm{d} \rho_{A_i R_i}$ and $\mathrm{Tr}_{A_i R_i}(\mathrm{d} \rho_{A_i R_i}) = 0$ by normalization condition $\mathrm{Tr}_{A_i R_i}(\rho_{A_i R_i}) = 1$.
    The differential of the reduced state $\rho_{A_i R_i}$ is
    \begin{align}
        & \mathrm{d} \rho_{A_i R_i} = \mathrm{Tr}_{\overline{A_i R_i}} \left(\mathrm{d} \rho_{\bigotimes_{i=1}^{n+1}A_i}\right) \nonumber \\
        & = \mathrm{i} \mathrm{Tr}_{\overline{A_i R_i}} \left([\mathrm{d}\Theta_{A_{n+1}}, \rho_{\bigotimes_{i=1}^{n+1}A_i}]\right) ,
    \end{align}
    and the differential of the target function is
    \begin{align}
        \mathrm{d} E = -\frac{\mathrm{i}}{2} \sum_{i=1}^n \mathrm{Tr}\left(\mathrm{d}\Theta_{A_{n+1}} [\rho_{\bigotimes_{i=1}^{n+1}A_i}, \log \rho_{A_i R_i}]\right) .
    \end{align}
    Thus, the gradient of the target function with respect to the Hermitian operator $\Theta_{A_{n+1}}$ is
    \begin{equation}
        \nabla_{\Theta_{A_{n+1}}} E = -\frac{\mathrm{i}}{2} \sum_{i=1}^n \mathrm{Tr}_{\bar{A}_{n+1}}\left([\rho_{\bigotimes_{i=1}^{n+1}A_i}, \log \rho_{A_i R_i}]\right) .
    \end{equation}

    The purification is updated as
    \begin{equation}
        \rho_{\bigotimes_{i=1}^{n+1}A_i} \leftarrow \hat{U}_{A_{n+1}} \rho_{\bigotimes_{i=1}^{n+1}A_i} \hat{U}_{A_{n+1}}^{\dagger},
    \end{equation}
    where the Hermitian operator is updated as
    \begin{equation}
        \Delta \Theta_{A_{n+1}} =  \frac{\mathrm{i}\alpha}{2} \sum_{i=1}^n \mathrm{Tr}_{\bar{A}_{n+1}}\left([\rho_{\bigotimes_{i=1}^{n+1}A_i}, \log \rho_{A_i R_i}]\right) ,
    \end{equation}
    where $\alpha >0$ is the learning rate, and $\bar{A}_{n+1} = \bigotimes_{i=1}^{n} A_i$.
    To optimize the learning rate, several methods can be used, such as the exact linear search and the Armijo linear search.

\section{Properties of generalized EoP and its gap} \label{app: properties}

\subsection{Lower and upper bounds of generalized EoP}

    For a $k$-partite mixed state $\rho_{\bigotimes_{i=1}^k A_i}$, consider its Schmidt decomposition
    \begin{equation}
        \rho_{\bigotimes_{i=1}^k A_i} = \sum_l p_l \ket{\psi_{\bigotimes_{i=1}^k A_i}^l}\bra{\cdot},
    \end{equation}
    where $\ket{\psi_{\bigotimes_{i=1}^k A_i}^l}$ are orthogonal pure states.
    Consider the purification of $\rho_{\bigotimes_{i=1}^k A_i}$
    \begin{equation}
        \ket{\psi_{\bigotimes_{i=1}^{k} A_i R}} = \sum_l \sqrt{p_l} \ket{\psi_{\bigotimes_{i=1}^{k} A_i}^l} \otimes \ket{l_{R}} .
    \end{equation}
    In the partition $(A_1|\dots|A_{k-1}|A_{k}R)$, we have
    \begin{align}
        & E_p(A_1,\dots, A_k)_{\rho} \leq \frac{1}{2} \left[\sum_{i=1}^{k-1} S(A_i)_{\psi} + S(A_k R)_{\psi}\right] \nonumber \\
        & = \frac{1}{2} \left[\sum_{i=1}^{k-1} S(A_i)_{\rho} + S\left(\bigotimes_{i=1}^{k-1} A_i\right)_{\rho}\right].
    \end{align}
    Therefore, we have the upper bound of the generalized EoP 
    \begin{align}
        & E_p(A_1,\dots, A_k) \nonumber \\
        & \leq \frac{1}{2} \min_{i} \left[\sum_{j\neq i} S(A_j) + S\left(\bigotimes_{j\neq i} A_j\right)\right],
    \end{align}
    which is attained if the state $\rho_{\bigotimes_{i=1}^k A_i}$ is pure,
    \begin{align} \label{eq: pure}
        & E_p(A_1,\dots, A_k) = \frac{1}{2} \sum_{i=1}^{k} S(A_i) \nonumber \\
        & = \frac{1}{2} \left[\sum_{j\neq i} S(A_j) + S\left(\bigotimes_{j\neq i} A_j\right)\right].
    \end{align}

    For the generalized EoP $E_p(A_1,\dots, A_k A_k')$ of the $k$-partite state $\rho_{\bigotimes_{i=1}^k A_i}$, assume the purification $\ket{\psi_{\bigotimes_{i=1}^{k} A_i A_k' A_{k+1}}}$ attains the minimum of $E_p(A_1,\dots, A_k A_k')$, where $A_{k+1} = \bigotimes_{i=1}^k R_i$.
    \begin{widetext}
    \begin{align}
        &E_p(A_1,\dots, A_k A_k') = \frac{1}{2} \sum_{i=1}^{k-1} S(A_i R_i) + \frac{1}{2} S(A_k A_k' R_k) \nonumber \\
        &= \frac{1}{2} \left[\sum_{i=1}^{k-1} S(A_i R_i) + S(A_k) - S\left(\bigotimes_{i=1}^{k-1} A_iR_i A_k'R_k\right)
        + S\left(\bigotimes_{i=1}^{k-1} A_i R_i\right) + S(A_k' R_k) - S\left(\bigotimes_{i=1}^{k-1} A_i R_i A_k\right) \right]\nonumber \\
        & = \frac{1}{2} \left[D\left(\rho_{\bigotimes_{i=1}^{k-1}A_iR_i A_k}\Big\Vert\bigotimes_{i=1}^{k-1} \rho_{A_iR_i} \otimes \rho_{A_k}\right) 
        + D\left(\rho_{\bigotimes_{i=1}^{k-1}A_iR_i A_k'R_k}\Big\Vert\rho_{\bigotimes_{i=1}^{k-1} A_iR_i} \otimes \rho_{A_k'R_k}\right)\right]\nonumber \\
        & \geq \frac{1}{2} \left[D\left(\rho_{\bigotimes_{i=1}^{k-1}A_i A_k}\Big\Vert\bigotimes_{i=1}^{k-1} \rho_{A_i} \otimes \rho_{A_k}\right) 
        + D\left(\rho_{\bigotimes_{i=1}^{k-1}A_i A_k'}\Big\Vert\rho_{\bigotimes_{i=1}^{k-1} A_i} \otimes \rho_{A_k'}\right)\right]\nonumber \\
        & = \frac{1}{2} \left[\sum_{i=1}^{k-1} S(A_i) + S(A_k) + S(A_k') + S\left(\bigotimes_{i=1}^{k-1} A_i\right) - S\left(\bigotimes_{i=1}^{k-1} A_i A_k\right) - S\left(\bigotimes_{i=1}^{k-1} A_i A_k'\right)\right] \nonumber \\
        & = \frac{1}{2} \left[\sum_{i=1}^{k} S(A_i) + S\left(\bigotimes_{i=1}^{k-1} A_i\right) - S\left(\bigotimes_{i=1}^{k-1} A_i \Big|A_k'\right) - S\left(\bigotimes_{i=1}^{k-1} A_i \Big|A_k'\right)\right].
    \end{align}
    \end{widetext}
    Here, in the second line, we use $S(A_k) = S\left(\bigotimes_{i=1}^{k-1} A_iR_i A_k'R_k\right)$ and $S(A_k' R_k) - S\left(\bigotimes_{i=1}^{k-1} A_i R_i A_k\right)$ for pure state $\ket{\psi_{\bigotimes_{i=1}^{k} A_i A_k' A_{k+1}}}$.
    In the forth line, we use the monotonicity of relative entropy under partial trace.
    In the last line, we use the conditional entropy $S(A|B) = S(AB) - S(B)$.

\subsection{Decadence of generalized EoP}

    Lemma~\ref{lemma: monotonicity} proves that generalized EoP is non-increasing under discarding subsystems.
    Moreover, for a decadent chain of subsystems $\alpha_{k} = \{A_{i_1}, A_{i_2}, \dots, A_{i_k}\} \subseteq \alpha_{k+1} = \{A_{i_1}, A_{i_2}, \dots, A_{i_k}, A_{i_{k+1}}\}$, we have 
    \begin{equation}
        E_p(\alpha_{k}) \leq E_p(\alpha_{k+1}).
    \end{equation}
    We prove this for the special case $\alpha_{k} = \{A_1, A_2, \dots, A_k\}$ and $\alpha_{k+1} = \{A_1, A_2, \dots, A_k, A_{k+1}\}$.
    Assume $E_p(\alpha_{k+1})$ is attained at $\ket{\psi_{\bigotimes_{i=1}^{k+1} A_i R_i}}$
    \begin{align}
        & E_p(A_1, \dots, A_{k+1}) = \frac{1}{2} \sum_{i=1}^{k+1} S(A_i R_i) \nonumber \\
        & \geq \frac{1}{2} \left[\sum_{i=1}^{k-1} S(A_i R_i) + S(A_k A_{k+1}R_kR_{k+1})\right] \nonumber \\
        & \geq E_p(A_1, \dots, A_k).
    \end{align}
    The first inequality is due to the fact $I(A_k R_k: A_{k+1} R_{k+1}) = S(A_k R_k) + S(A_{k+1} R_{k+1}) - S(A_k A_{k+1}R_kR_{k+1}) \geq 0$, and the second inequality is due to the definition of $E_p(A_1, \dots, A_k)$.

\subsection{Subadditivity}

    The generalized EoP is subadditive 
    \begin{align}
        & E_p(A_1, \dots, A_k) + E_p(A_1',\dots,A_k') \nonumber \\
        & = \frac{1}{2} \sum_{i=1}^{n} \left[S(A_iR_i) + S(A_i'R_i')\right] \geq \frac{1}{2} \sum_{i=1}^{n} S(A_i A_i'R_iR_i') \nonumber \\
        & \geq E_p(A_1A_1',\dots,A_kA_k').
    \end{align}

\subsection{Polygamy}

    For pure state $\rho_{\bigotimes_{i=1}^{k-1} A_i A_k A_k'}$, the generalized EoP is polygamous.
    In Eq.~(\ref{eq: polygamy}), the second line is followed from the lower bound of generalized EoP, Eq.~(\ref{eq: lower_bound}), the first equality in third line is followed since $\rho_{\bigotimes_{i=1}^{k-1} A_i A_k A_k'}$ is pure, the second inequality since $S(A) + S(B) \geq S(AB)$, and the last equality is followed from Eq.~(\ref{eq: pure}).
    \begin{widetext}
    \begin{align} \label{eq: polygamy}
        & E_p(A_1,\dots,A_{k-1},A_k) + E_p(A_1,\dots,A_{k-1},A_k') \nonumber \\
        & \geq \sum_{i=1}^{k-1} S(A_i) + \frac{1}{2}\left[S(A_k) + S(A_k') - S\left(\bigotimes_{i=1}^{k-1} A_i A_k\right) - S\left(\bigotimes_{i=1}^{k-1} A_i A_k'\right)\right] \nonumber \\
        & = \sum_{i=1}^{k-1} S(A_i) \geq \frac{1}{2} \left[\sum_{i=1}^{k-1}S(A_i) + S\left(\bigotimes_{i=1}^{k-1} A_i\right)\right]  = E_p(A_1,\dots,A_{k-1},A_kA_k').
    \end{align}
    \end{widetext}

\subsection{Lower and upper bounds of generalized EoP gap for convex combination}

    Assume a mixed state $\rho_{\bigotimes_{i=1}^k A_i}$ is the convex combination of states $\rho_{\bigotimes_{i=1}^k A_i}^l$
    \begin{equation}
        \rho_{\bigotimes_{i=1}^k A_i} = \sum_l p_l \rho_{\bigotimes_{i=1}^k A_i}^l.
    \end{equation}
    Let $\ket{\psi_{\bigotimes_{i=1}^{k+1} A_i}^l}$ be the purification of $\rho_{\bigotimes_{i=1}^k A_i}^l$ that attains the minimum of $E_p(A_1,\dots,A_k)$, where $A_{k+1} = \bigotimes_{i=1}^k R_i$.
    Consider the purification of $\rho_{\bigotimes_{i=1}^k A_i}$
    \begin{equation}
        \ket{\psi_{\bigotimes_{i=1}^{k} A_i \tilde{A}_{k+1}}} = \sum_l \sqrt{p_l} \bigotimes_{i=1}^{k} \ket{l_{\tilde{R}_i}} \otimes \ket{\psi_{\bigotimes_{i=1}^{k+1} A_i}^l} ,
    \end{equation}
    where $\tilde{A}_{k+1} = \bigotimes_{i=1}^k \tilde{R}_i \otimes A_{k+1}$.
    The generalized EoP of $\rho_{\bigotimes_{i=1}^k A_i}$ in partition $\alpha_k = \{A_1,A_2,\dots,A_k\}$ is upper bounded by
    \begin{align}
        & E_p(\alpha_k)_{\rho} \leq \frac{1}{2} \sum_{i=1}^k S(A_i R_i \tilde{R}_i)_{\psi} \nonumber \\
        & = \frac{1}{2} \sum_{i=1}^k \left[\sum_l p_l S(A_i R_i)_{\psi^l} + H(p_l)\right] \nonumber \\
        & = \sum_l p_l E_p(\alpha_k)_{\rho^l} + \frac{k}{2} H(p_l).
    \end{align}
    The lower bound of $2 E_p(\alpha_k)$ of $\rho_{\bigotimes_{i=1}^k A_i}$ has the lower bound
    \begin{align}
        & \sum_i S(A_i)_{\rho} - S(\tilde{A}_{k+1})_{\rho} \nonumber \\
        & \geq \sum_i \sum_l p_l S(A_i)_{\rho^l} 
        - \sum_l p_l S(A_{k+1})_{\rho^l} - H(p_l) \nonumber \\
        & = \sum_l p_l \left[\sum_i S(A_i)_{\rho^l} - S(A_{k+1})_{\rho^l}\right] - H(p_l).
    \end{align}
    Thus, the upper bound of the generalized EoP gap of $\rho_{\bigotimes_{i=1}^k A_i}$ is
    \begin{equation}
        g(\alpha_k)_{\rho} \leq \sum_l p_l g(\alpha_k)_{\rho^l} + \frac{k+1}{2} H(p_l).
    \end{equation}

    Moreover, we can also derive a lower bound of the generalized EoP gap of the mixed state $\rho_{\bigotimes_{i=1}^k A_i}$.
    The generalized EoP of $\rho_{\bigotimes_{i=1}^k A_i}$ in partition $\alpha_k = \{1,2,\dots,k\}$ is lower bounded by 
    \begin{align}
        & E_p(\alpha_k)_{\rho} = \frac{1}{2} \sum_{i=1}^k S(A_i R_i \tilde{R}_i)_{\hat{U}\psi} \nonumber \\
        & \geq \frac{1}{2} \sum_{i=1}^k \sum_l p_l S(A_i R_i \tilde{R}_i)_{\hat{U}\psi^l} \geq \sum_l p_l E_p(\alpha_k)_{\rho^l},
    \end{align}
    where $\hat{U}$ is the unitary operation on system $\tilde{A}_{k+1}$ that attains the minimum of $E_p(\alpha_k)$.
    The lower bound of $2 E_p(\alpha_k)$ of $\rho_{\bigotimes_{i=1}^k A_i}$ has the upper bound
    \begin{align}
        & \sum_i S(A_i)_{\rho} - S(\tilde{A}_{k+1})_{\rho} \nonumber \\
        & \leq \sum_i \left[\sum_l p_l S(A_i)_{\rho^l} + H(p_l)\right] 
        - \sum_l p_l S(A_{k+1})_{\rho^l} \nonumber \\
        & = \sum_l p_l \left[\sum_i S(A_i)_{\rho^l} - S(A_{k+1})_{\rho^l}\right] + k H(p_l).
    \end{align}
    Thus, the lower bound of the generalized EoP gap of $\rho_{\bigotimes_{i=1}^k A_i}$ is
    \begin{equation}
        g(\alpha_k)_{\rho} \geq \sum_l p_l g(\alpha_k)_{\rho^l} - \frac{k}{2} H(p_l).
    \end{equation}